\documentclass[letterpaper,twocolumn,twoside]{ieeetran}

\usepackage{amsmath,amssymb,mathtools}
\usepackage{graphicx,color,subfig}
\usepackage{cite}
\usepackage[font=footnotesize,labelfont=bf]{caption}

\usepackage{amsthm}
\let\labelindent\relax
\usepackage{tikz,calc,bbding}
\usetikzlibrary{shapes, arrows, decorations, decorations.pathreplacing, decorations.pathmorphing, decorations.markings, fit, matrix}
\usepackage{arydshln}
\usepackage[hidelinks]{hyperref}

\graphicspath{{figures/}}

\newtheorem{theorem}{Theorem}[section]

\newtheorem{lemma}[theorem]{Lemma}

\newtheorem{remark}[theorem]{Remark}

\newtheorem{example}[theorem]{Example}

\newcommand{\longthmtitle}[1]{\mbox{}{\bf \textit{(#1).}}}

\newcommand{\real}{\ensuremath{\mathbb{R}}}
\newcommand{\realpos}{\ensuremath{\mathbb{R}_{>0}}}
\newcommand{\realnonneg}{\ensuremath{\mathbb{R}_{\ge 0}}}
\newcommand{\realnonpos}{\ensuremath{\mathbb{R}_{\le 0}}}

\newcommand{\setdef}[2]{\{#1 \; | \; #2\}}
\newcommand{\setdefb}[2]{\big\{#1 \; | \; #2\big\}}

\newcommand{\matlab}[1]{{\fontfamily{phv}\selectfont \small{#1}}}

\newcommand{\Nc}{\mathcal{N}}

\newcommand\abf{\mathbf{a}}
\newcommand\bbf{\mathbf{b}}
\newcommand\cbf{\mathbf{c}}
\newcommand\dbf{\mathbf{d}}
\newcommand\fbf{\mathbf{f}}
\newcommand\gbf{\mathbf{g}}
\newcommand\mbf{\mathbf{m}}
\newcommand\nbf{\mathbf{n}}
\newcommand\pbf{\mathbf{p}}

\newcommand\ubf{\mathbf{u}}
\newcommand\vbf{\mathbf{v}}
\newcommand\wbf{\mathbf{w}}
\newcommand\xbf{\mathbf{x}}

\newcommand\Abf{\mathbf{A}}
\newcommand\Bbf{\mathbf{B}}

\newcommand\Fbf{\mathbf{F}}
\newcommand\Gbf{\mathbf{G}}
\newcommand\Ibf{\mathbf{I}}
\newcommand\Kbf{\mathbf{K}}
\newcommand\Mbf{\mathbf{M}}
\newcommand\Nbf{\mathbf{N}}

\newcommand\Wbf{\mathbf{W}}
\newcommand\alphab{{\boldsymbol{\alpha}}}

\newcommand\sigmab{{\boldsymbol{\sigma}}}
\newcommand\Sigmab{\boldsymbol{\Sigma}}

\newcommand\xib{\boldsymbol{\xi}}
\newcommand\nub{\boldsymbol{\nu}}

\newcommand\epsilonb{\boldsymbol{\epsilon}}

\newcommand{\ones}{\mathbf{1}}
\newcommand{\zeros}{\mathbf{0}}

\newcommand{\until}[1]{\{1,\dots,#1\}}

\newcommand\s{{\rm s}}
\newcommand\zls{\{0, \ell, \s\}}

\newcommand\ssm{{\raisebox{2pt}{\scriptsize
      $\mathfrak{0}$}}}
\newcommand\ssp{{\raisebox{2pt}{\scriptsize
      $\mathfrak{1}$}}}
\newcommand\ssmp{{\raisebox{1pt}{\scriptsize $\mathfrak{\text{all}}$}}}

\usepackage{enumitem}
\renewcommand{\theenumi}{(\roman{enumi})}

\renewcommand{\footnoterule}{%
  \hspace{3pt} \hrule width 0.4\textwidth height 0.5pt
  \kern 2pt
}

\makeatletter
\renewcommand*\env@matrix[1][\arraystretch]{%
  \edef\arraystretch{#1}%
  \hskip -\arraycolsep
  \let\@ifnextchar\new@ifnextchar
  \array{*\c@MaxMatrixCols c}}
\makeatother

\newcommand{\oprocendsymbol}{\hbox{$\square$}}
\newcommand{\oprocend}{\relax\ifmmode\else\unskip\hfill\fi\oprocendsymbol}

\renewcommand{\baselinestretch}{.99}

\tikzset{->-/.style={decoration={
  markings,
  mark=at position #1 with {\arrow{>}}},postaction={decorate}}}
  
\title{\Large \bf Hierarchical Selective Recruitment in
  Linear-Threshold Brain Networks
  \\
  Part II: Multi-Layer Dynamics and Top-Down Recruitment%
  \thanks{A preliminary version appeared as~\cite{EN-JC:18-cdc} at the
    IEEE Conference on Decision and Control}}

\author{Erfan Nozari \quad Jorge Cort\'es%
  \thanks{Erfan Nozari is with the Department of Electrical and Systems Engineering, University of Pennsylvania, enozari@seas.upenn.edu. Jorge Cort\'es is with the Department of
    Mechanical and Aerospace Engineering, University of California, San Diego,
    cortes@ucsd.edu.}}

\begin{document}

\maketitle
\thispagestyle{empty}
\pagestyle{empty}

\begin{abstract}
  Goal-driven selective attention (GDSA) is a remarkable function that
  allows the complex dynamical networks of the brain to support
  coherent perception and cognition. Part I of this two-part paper
  proposes a new control-theoretic framework, termed hierarchical
  selective recruitment (HSR), to rigorously explain the emergence of
  GDSA from the brain's network structure and dynamics.  This part
  completes the development of HSR by deriving conditions on the joint
  structure of the hierarchical subnetworks that guarantee top-down
  recruitment of the task-relevant part of each subnetwork by the
  subnetwork at the layer immediately above, while inhibiting the
  activity of task-irrelevant subnetworks at all the hierarchical
  layers. To further verify the merit and applicability of this
  framework, we carry out a comprehensive case study of selective
  listening in rodents and show that a small network with HSR-based
  structure and minimal size can explain the data with remarkable
  accuracy while satisfying the theoretical requirements of HSR. Our
  technical approach relies on the theory of switched systems and
  provides a novel converse Lyapunov theorem for state-dependent
  switched affine systems that is of independent interest.
\end{abstract}

\section{Introduction}\label{sec:intro}

Our ability to construct a dynamic yet coherent perception of the
world, despite the numerous parallel sources of information that
affect our senses, is to a great extent reliant on the brain's
capability to prioritize the processing of task-relevant information
over task-irrelevant distractions according to one's goals and
desires. This capability, commonly known as goal-driven selective
attention (GDSA), has been the subject of extensive research over the
past decades. Despite major advances, a fundamental understanding of
GDSA and, in particular, how it emerges from the dynamics of the
underlying neuronal networks, is still missing. The aim of this work
is to reduce this gap by bringing tools and insights from systems and
control theory into these questions from neuroscience.

In this two-part paper, we propose the novel theoretical framework of
Hierarchical Selective Recruitment (HSR) for GDSA. As stated in Part
I, HSR consists of a novel hierarchical model of brain organization, a
set of analytical results regarding the multi-timescale dynamics of
this model, and a careful translation between the properties of these
dynamics and well known experimental observations about GDSA. Inspired
and supported by extensive experimental
research~\cite{ECC:53,AMT:69,JM-RD:85,BCM:93,RD-JD:95,SK-PD-RD-LGU:98,LI-CK:01,MAP-GMD-SK:04,NL:05,JJF-ACS:11,AG-ACN:12,CCR-MRD:14,MG-KH-EN:16}, %
HSR relies on four major assumptions about the neuronal mechanisms
underlying GDSA. These include (i) the brain's hierarchical
organization, so that (cognitively-)higher areas provide control
inputs to the activity of lower level
ones~\cite{RD-JD:95,LI-CK:01,MAP-GMD-SK:04,NL:05,AG-ACN:12,CCR-MRD:14,MG-KH-EN:16}, %
(ii) its sparse coding, so that task-relevant and task-irrelevant
stimuli is represented and processed by sufficiently distinct neuronal
populations (particularly if the two stimuli differ in major or
multiple properties, such as location, sensory modality,
etc.)~\cite{JM-RD:85,BCM:93,RD-JD:95,SK-PD-RD-LGU:98,LI-CK:01,MAP-GMD-SK:04,AG-ACN:12,MG-KH-EN:16},
(iii) the distributed and graded nature of GDSA, so that selective
attention happens at multiple layers of the
hierarchy~\cite{AMT:69,BCM:93,RD-JD:95,SK-PD-RD-LGU:98,LI-CK:01,MAP-GMD-SK:04,JJF-ACS:11,AG-ACN:12,MG-KH-EN:16},
and (iv) the concurrence of the suppression and enhancement of
task-irrelevant and task-relevant activity,
resp.~\cite{ECC:53,AMT:69,JM-RD:85,BCM:93,RD-JD:95,SK-PD-RD-LGU:98,MAP-GMD-SK:04,NL:05,JJF-ACS:11,AG-ACN:12,CCR-MRD:14,MG-KH-EN:16} %
(formulated as \emph{selective inhibition} and \emph{top-down
  recruitment} in HSR, resp.).

The hierarchical structure of the brain plays a key role in both
selective inhibition and top-down recruitment. The position of brain
areas along this hierarchy is determined based on the direction in
which sensory information and decisions flow, but also by the
separation of timescales between the areas.  As expected, the
timescale of the internal dynamics of the neuronal populations
increases (becomes slower) as one moves up the
hierarchy~\cite{UH-EY-IV-DJH-NR:08,CJH-TT-THD-LJS-CEC-OD-WKD-NR-DJH-UH:12,BG-EE-GH-AG-AK:12,UH-JC-CJH:15,MGM-DAK-SLT-GKA:16,JDM-AB-DJF-RR-JDW-XC-CP-TP-HS-DL-XW:14,RC-KK-MG-HK-XW:15}.
Although this hierarchy of timescales is well known in neuroscience,
its role in GDSA has remained, to the best of our knowledge,
uncharacterized. Using tools from singular perturbation theory, we
here reveal the critical role played by this separation of timescales
in the top-down recruitment of the task-relevant subnetworks and
provide rigorous conditions on the joint structure of all layers that
guarantee such recruitment.

\subsubsection*{Literature Review}

The hierarchical organization of the brain has been recognized for
decades~\cite{NT:50,ARL:70,DJF-DCEV:91} and applies to multiple
aspects of brain structure and function. These aspects include (i)
network
topology~\cite{DJF-DCEV:91,AK-ATR-EW-GB-RK:10,GZL-CZ-JK:10,NTM-JV-PC-AF-RQ-CH-CL-PM-PG-SU-PB-CD-KK-HK:14}
(where nodes are assigned to layers based on their position on
bottom-up and top-down pathways), (ii) encoding
properties~\cite{PL:98,DB-MD:09} (where nodes that have larger
receptive fields and/or encode more abstract stimulus properties
constitute higher layers), (iii) dynamical
timescale~\cite{UH-EY-IV-DJH-NR:08,CJH-TT-THD-LJS-CEC-OD-WKD-NR-DJH-UH:12,BG-EE-GH-AG-AK:12,UH-JC-CJH:15,MGM-DAK-SLT-GKA:16,JDM-AB-DJF-RR-JDW-XC-CP-TP-HS-DL-XW:14,RC-KK-MG-HK-XW:15,AK-ATR-EW-GB-RK:10,NTM-JV-PC-AF-RQ-CH-CL-PM-PG-SU-PB-CD-KK-HK:14,JPG-MM-JIL:16,CC-HA-DB-YB-SM:14,CAR-EP-SP-CDH:17,SJK-JD-KJF:08,YY-JT:08}
(where nodes are grouped into layers according to the timescale of
their dynamics), (iv) nodal
clustering~\cite{DSB-ETB-BAV-VSM-DRW-AML:08,DM-RL-AF-KE-ETB:09,DM-RL-ETB:10,ZZ-HF-JL:13}
(where nodes only constitute the leafs of a clustering tree), and (v)
oscillatory activity~\cite{PL-ASS-KHK-IU-GK-CES:05} (where layers
correspond to nested oscillatory frequency bands). Note that while
hierarchical layers are composed of brain regions in (i)-(iii), this
is not the case for (iv) and (v). The hierarchies (i)-(iii) are
remarkably similar (in terms of the assignment of brain regions to the
layers of the hierarchy), and here we particularly focus on (iii) (the
timescale separation between hierarchical layers) as it plays a
pivotal role in HSR.

Studies of timescale separation between cortical regions are more
recent. Several experimental works have demonstrated a clear increase
in intrinsic timescales as one moves up the hierarchy using indirect
measures such as the length of stimulus presentation that elicits a
response~\cite{UH-EY-IV-DJH-NR:08,CJH-TT-THD-LJS-CEC-OD-WKD-NR-DJH-UH:12},
resonance frequency~\cite{BG-EE-GH-AG-AK:12}, the length of the
largest time window over which the responses to successive stimuli
interfere~\cite{UH-JC-CJH:15}, and how quickly the activation level of
any brain region can track changes in sensory
stimuli~\cite{MGM-DAK-SLT-GKA:16}. Direct evidence for this
hierarchical separation of timescales was indeed provided
in~\cite{JDM-AB-DJF-RR-JDW-XC-CP-TP-HS-DL-XW:14} using the decay rate
of spike-count autocorrelation functions. This was shown even more
comprehensively in~\cite{RC-KK-MG-HK-XW:15} using linear-threshold
rate models and the concept of \emph{continuous
  hierarchies}~\cite{AK-ATR-EW-GB-RK:10,NTM-JV-PC-AF-RQ-CH-CL-PM-PG-SU-PB-CD-KK-HK:14}
(whereby the layer of each node can vary continuously according to its
intrinsic timescale, therefore removing the rigidity and arbitrariness
of node assignment in classical hierarchical
structures). Interestingly, recent studies show that this timescale
variability may have roots not only in synaptic dynamics of individual
neurons~\cite{JPG-MM-JIL:16}, but also in sub-neuronal genetic
factors~\cite{CC-HA-DB-YB-SM:14} as well as supra-neuronal network
structures~\cite{CAR-EP-SP-CDH:17}. In terms of applications,
computational models of motor control were perhaps the first to
exploit this cortical hierarchy of
timescales~\cite{SJK-JD-KJF:08,YY-JT:08}. Despite the vastness of the
literature on its roots and applications, we are not aware of any
theoretical analysis of the effects of this separation of timescales
on the hierarchical dynamics of neuronal networks.

The accompanying Part I~\cite{EN-JC:18-tacI} proposes the HSR
framework, which is strongly rooted in this separation of timescales.
Part I analyzes the internal dynamics of the subnetworks at each layer
of the hierarchy. Using the class of linear-threshold network models,
it characterizes the networks that have a unique equilibrium, are
locally/globally asymptotically stable, and have bounded
trajectories. In Part I, we also provide a detailed account of
feedforward and feedback mechanisms for selective inhibition between
any two layers of the hierarchy and show that the internal dynamical
properties of the task-relevant subnetwork at each layer is the sole
determiner of the dynamical properties achievable under selective
inhibition.

In this paper, we complete the development of the HSR framework for
GDSA by analyzing the mechanisms for top-down recruitment of the
task-relevant subnetwork, combining it with the feedforward and
feedback mechanisms of selective inhibition, and generalizing the
combination to arbitrary number of layers.  Top-down recruitment is
one of the most experimentally well-documented phenomena in selective
attention, see,
e.g.,~\cite{JM-RD:85,BCM:93,RD-JD:95,SK-PD-RD-LGU:98,LI-CK:01,MAP-GMD-SK:04,AG-ACN:12,CCR-MRD:14,MG-KH-EN:16}. While
the enhancement (a.k.a. \emph{modulation}) of activity in the
task-relevant populations is the simplest form of recruitment, our
model is general and thus also allows for more complex observed forms
of recruitment, such as changes in the receptive fields%
\footnote{The receptive field of each neuron is the area in the
  stimulus space where the neuron is responsive to the presence of
  stimuli.\label{ft:rf}}%
~\cite{YY-MC:98,JBF-ME-SVD-SAS:07,KA-VMS-ST:09}.

In the analysis of top-down recruitment, we use tools from singular
perturbation theory to rigorously leverage this separation of
timescales. The classical result on singularly perturbed ODEs goes
back to Tikhonov~\cite{ANT:52}, \cite[Thm 11.1]{HKK:02} and has since
inspired an extensive literature, see,
e.g.~\cite{ABV:94,DN:02,JKK-JDC:12,REO:12}. Tikhonov's result,
however, requires smoothness of the vector fields, which is not
satisfied by linear-threshold dynamics. Fortunately, several works
have sought extensions to non-smooth differential equations and even
differential
inclusions~\cite{ALD-VMV:83,ALD-IIS:90,MQ:95,AD-TD-IIS:96},
culminating in the work~\cite{VV:97} which we use here. Similar to
Tikhonov's original work, \cite{VV:97} only applies to finite
intervals. Extensions to infinite intervals exist~\cite{FW:05,GG:05}
but, as expected, they require asymptotic stability of the
reduced-order model (ROM) which we do not in general have.%
\footnote{Recall that in two-timescale dynamics, ROM results from
  replacing the fast variable with its equilibrium (reducing order to
  that of the slow variable).}

\subsubsection*{Statement of Contributions}

The paper has four main contributions. First, we use the timescale
separation in hierarchical brain networks and the theory of singular
perturbations to provide an analytic account of top-down recruitment
in terms of conditions on the network structure. These conditions
guarantee the stability of the task-relevant part of a (fast)
linear-threshold subnetwork towards a reference trajectory set by a
slower subnetwork. This, in particular, subsumes the most classical
enhancement (strengthening) of the activity of task-relevant nodes but
is more general and can account for recent, complex observations such
as changes in neuronal receptive fields under
GDSA~$^{\!\!\ref{ft:rf}}$. We further combine these results with the
results of Part I to allow for simultaneous selective inhibition and
top-down recruitment, as observed in GDSA.  Second, we extend this
combination to hierarchical structures with an arbitrary number of
layers, as observed in nature, to yield a fully developed HSR
framework.  Here, we also derive an extension of the stability results
in Part I that guarantees GES of a multi-layer multiple timescale
linear-threshold network. Third, to validate the proposed HSR
framework, we provide a detailed case study of GDSA in real brain
networks. Using single-unit recordings from two brain regions of
rodents performing a selective listening task, we provide an in-depth
analysis of appropriate choices of neuronal populations in each brain
region as well as the timescales of their dynamics. We propose a novel
hierarchical structure for these populations, tune the parameters of
the resulting network using a novel objective function, and show that
the resulting structure conforms to the theoretical results and
requirements of HSR while explaining more than $90\%$ of variability
in the data. As part of our technical approach, our fourth and final
contribution is a novel converse Lyapunov theorem that extends the
state of the art on GES for state-dependent switched affine
systems. This result only requires continuity of the vector field and
guarantees the existence of an infinitely smooth quadratically-growing
Lyapunov function if the dynamics is GES. Because of independent
interest, we formulate and prove the result for general
state-dependent switched affine systems.%
\footnote{Throughout the paper, we use the following notation.
  $\real$, $\realnonneg$, $\realnonpos$, and $\realpos$ denote the set
  of reals, nonnegative reals, nonpositive reals, and positive reals,
  resp.  $\ones_n$, $\zeros_n$, $\zeros_{p \times n}$, and $\Ibf_n$
  stand for the $n$-vector of all ones, the $n$-vector of all zeros,
  the $p$-by-$n$ zero matrix, and the identity $n$-by-$n$ matrix,
  resp. The subscripts are omitted when clear from the context.  When
  a vector $\xbf$ or matrix $\Abf$ are block-partitioned, $\xbf_i$ and
  $\Abf_{ij}$ refer to the $i$th block of $\xbf$ and $(i, j)$th block
  of $\Abf$, resp. Given $\Abf \in \real^{n \times n}$, its
  element-wise absolute value and spectral radius are $|\Abf|$ and
  $\rho(\Abf)$, resp.  $\|\cdot\|$ denotes vector $2$-norm.  For $x
  \in \real$ and $m \in \realpos \cup \{\infty\}$, %
  $[x]_0^m = \min\{\max\{x, 0\}, m\}$, which is the projection of $x$
  onto $[0, m]$. When $\xbf \in \real^n$ and $\mbf \in \realpos^n \cup
  \{\infty\}^n$, we similarly define $[\xbf]_\zeros^\mbf =
  [[x_1]_0^{m_1} \quad \cdots \quad [x_n]_0^{m_n}]^T$. For $\sigmab
  \in \zls^n$, $\Sigmab^\ell = \Sigmab^\ell(\sigmab)$ is a diagonal
  matrix with $\Sigma^\ell_{ii} = 1$ if $\sigma_i = \ell$ and
  $\Sigma^\ell_{ii} = 0$ otherwise. $\Sigmab^\s$ is defined
  similarly. We set the convention that $\Sigmab^\s \mbf = \zeros$ if
  $\Sigmab^\s = \zeros$ and $\mbf = \infty \ones_n$. For $D \subseteq
  \real^n$ and $\Abf \in \real^{p \times n}, \bbf \in \real^p$, we let
  $\Abf D + \bbf = \setdef{\Abf \xbf + \bbf}{\xbf \in D}$.}

\begin{figure}[tbh]
  \begin{center} 
    \begin{tikzpicture} 
          \node[circle, draw, line width=0.6pt, inner sep=-0.4pt] (1) {\includegraphics[width=16pt]{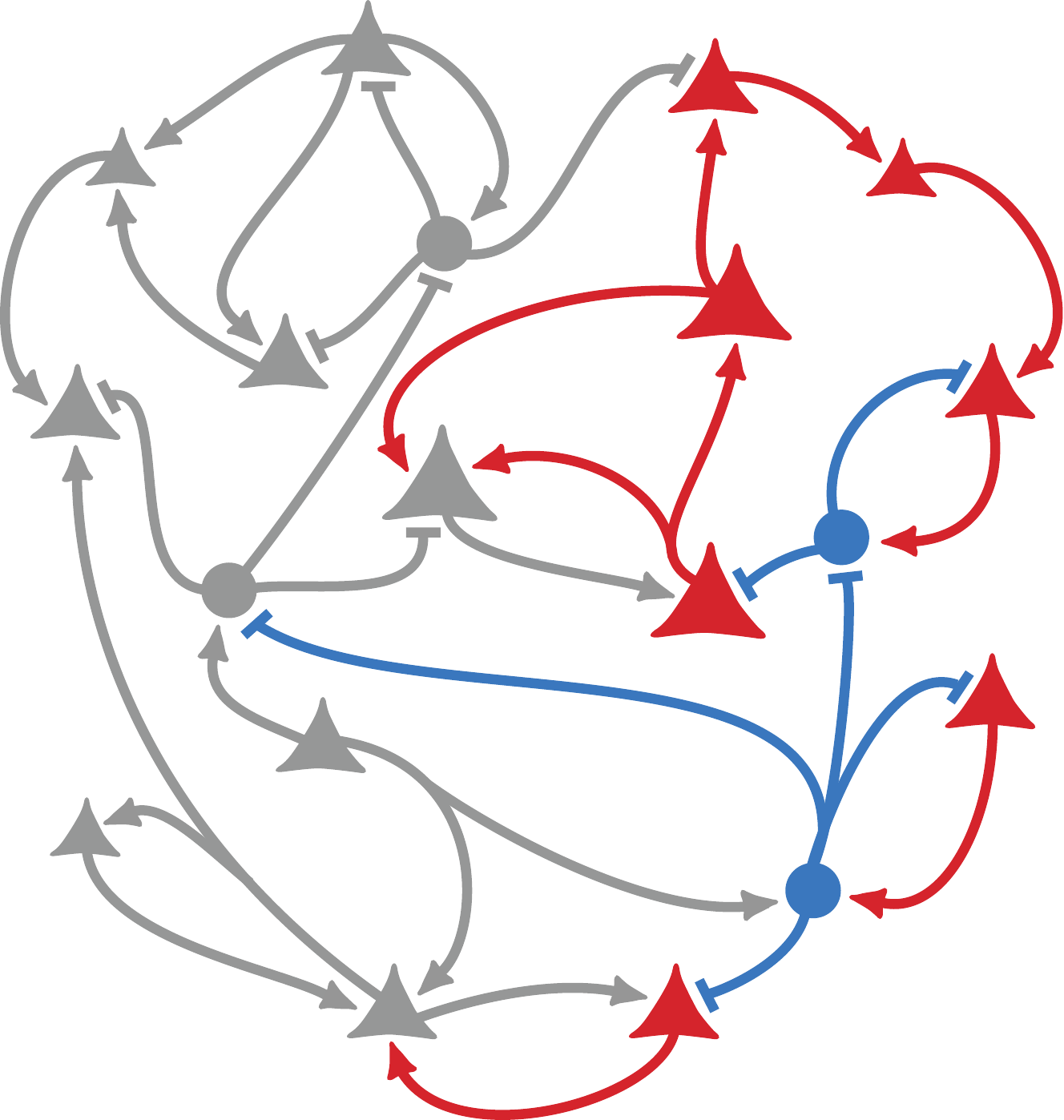}};
          \draw[-, line width=0.3pt] (1.north) to (1.south);
          \node[below of=1, yshift=-10pt, circle, draw, line width=0.6pt, inner sep=-0.4pt] (2) {\includegraphics[width=16pt]{EI_main_half_gray}};
          \draw[-, line width=0.3pt] (2.north) to (2.south);
          \node[below of=2, yshift=-10pt, circle, draw, line width=0.6pt, inner sep=-0.4pt] (3) {\includegraphics[width=16pt]{EI_main_half_gray}};
          \draw[-, line width=0.3pt] (3.north) to (3.south);
          \draw[latex-latex, line width=0.3pt, shorten <=1pt, shorten >=1pt] (1.270) to (2.90);
          \draw[latex-latex, line width=0.3pt, shorten <=1pt, shorten >=1pt] (2.270) to (3.90);
          \node[below of=3, yshift=10pt] {$\vdots$};
          \node[above of=1, yshift=-5pt] {$\vdots$};
      \node[left of=2, xshift=-20pt, yshift=40pt, scale=0.8] (i-1) {Subnetwork $i - 1$};
      \node[below of=i-1, yshift=-11pt, scale=0.8] (i) {Subnetwork $i$};
      \node[below of=i, yshift=-12pt, scale=0.8] (i+1) {Subnetwork $i + 1$};
      \node[right of=2, xshift=50pt, circle, draw, line width=1pt, inner sep=-2pt] (i-big) {\includegraphics[width=40pt]{EI_main_half_gray}};
      \draw[-] (i-big.north) to (i-big.south);
      \draw[shorten <=5pt, shorten >=25pt] (2.75) to (i-big.85);
      \draw[shorten <=5pt, shorten >=25pt] (2.285) to (i-big.275);
      \node[below of=i-big, xshift=-30pt, yshift=-25pt, scale=0.8] (i1) {\parbox{50pt}{\centering $\Nc_i^\ssm$ \\ (inhibited, state = $\xbf_i^\ssm$)}};
      \node[below of=i-big, xshift=30pt, yshift=-25pt, scale=0.8] (i2) {\parbox{50pt}{\centering $\Nc_i^\ssp$ \\ (recruited, state = $\xbf_i^\ssp$)}};
      \draw[-latex, bend right=0, shorten <=-8pt, shorten >=3pt] (i1.45) to (i-big.250);
      \draw[-latex, bend right=0, shorten <=-8pt, shorten >=3pt] (i2.135) to (i-big.290);
    \end{tikzpicture} 
  \end{center} 
  \caption{Hierarchical network structure considered in this
    work.}\label{fig:multi}
\vspace*{-1.5ex}
\end{figure}

\section{Problem Statement}\label{sec:prob-state}

The problem formulation is the same as in Part
I~\cite{EN-JC:18-tacI}. We include here a streamlined description for
a self-contained exposition.  We consider a hierarchical neuronal
network $\Nc$, cf. Figure~\ref{fig:multi}, whereby the nodes in each
layer $\Nc_i$ are further decomposed into a task-irrelevant part
$\Nc_i^\ssm$ and a task-relevant part $\Nc_i^\ssp$. The state
evolution of each layer $\Nc_i$ is modeled with linear-threshold
network dynamics of the form
\begin{align}\label{eq:dyn}
  \notag \tau_i \dot \xbf_i(t) = -\xbf_i(t) + [\Wbf_{i, i} \xbf_i(t) +
  \dbf_i&(t)]_\zeros^{\mbf_i}, \ \zeros \le \xbf_i(0) \le \mbf_i,
  \\
  &\mbf_i \in \realpos^n \cup \{\infty\}^n, \!\!
\end{align}
where $\xbf_i \in \real^{n_i}$, $\Wbf_{i, i} \in \real^{n_i \times
  n_i}$, $\dbf_i \in \real^{n_i}$, and $\mbf \in \realpos^n$ denote
the state, internal synaptic connectivity, external inputs, and state
upper bounds of $\Nc_i$, resp.%
\footnote{We note that this is a standard and widely used model in
  computational neuroscience, as mentioned in Part
  I~\cite{EN-JC:18-tacI}. Please see therein for a detailed literature
  review of its origins and prior analysis.}

The development of HSR is structured in four thrusts:
\begin{enumerate}[wide]
\item the analysis of the relationship between structure ($\Wbf_{i,
    i}$) and dynamical behavior
  for each subnetwork when operating in isolation from the rest of the
  network ($\dbf_i(t) \equiv \dbf_i$);
\item the analysis of the conditions on the joint structure of each
  two successive layers $\Nc_i$ and $\Nc_{i+1}$ that allows for
  selective inhibition of $\Nc_{i+1}^\ssm$ by its input from $\Nc_i$,
  being equivalent to the stabilization of $\Nc_{i+1}^\ssm$ to the
  origin %
  (inactivity);
\item the analysis of the conditions on the joint structure of each
  two successive layers $\Nc_i$ and $\Nc_{i+1}$ that allows for
  top-down recruitment of $\Nc_{i+1}^\ssp$ by its input from $\Nc_i$,
  being equivalent to the stabilization of $\Nc_{i+1}^\ssp$ toward a
  desired trajectory set by $\Nc_i$ (activity);
\item the combination of (ii) and (iii) in a unified framework and its
  extension to the complete $N$-layer network of networks.
\end{enumerate} 
Problems (i) and (ii) are addressed in Part I~\cite{EN-JC:18-tacI},
while problems (iii) and (iv) are the subject of this work.

We let
\begin{align}
  \dbf_i(t) = \Bbf_i \ubf_i(t) + \tilde \dbf_i(t), \qquad \ubf_i \in
  \realnonneg^{p_i},
\end{align} 
where $\ubf_i$ is the top-down control used
for inhibition of $\Nc_i^\ssm$. While in Part I we assumed for
simplicity that $\tilde \dbf_i(t)$ is given and constant, we here
consider its complete form
\begin{align*}
  \tilde \dbf_i(t) = \Wbf_{i, i-1} \xbf_{i-1}(t) + \Wbf_{i, i+1}
  \xbf_{i+1}(t) + \cbf_i,
\end{align*}
where the inter-layer connectivity matrices $\Wbf_{i, i-1}$ and
$\Wbf_{i, i+1}$ have appropriate dimensions and $\cbf_i \in
\real^{n_i}$ captures un-modeled background activity and possibly
nonzero activation thresholds. Substituting these into~\eqref{eq:dyn},
the dynamics of each layer $\Nc_i$ is given by
{\interdisplaylinepenalty=10000
  \begin{align}\label{eq:dyn-multi} \tau_i \dot \xbf_i(t) =
    -\xbf_i(t) + [&\Wbf_{i,i} \xbf_i(t) + \Wbf_{i, i-1} \xbf_{i-1}(t)
    \\
    \notag &+ \Wbf_{i, i+1} \xbf_{i+1}(t) + \Bbf_i \ubf_i(t) +
    \cbf_i]_\zeros^{\mbf_i}.
  \end{align}
}
Also following Part I, we partition network variables as %
\begin{align}\label{eq:wbi} 
  \notag \xbf_i = \begin{bmatrix}[1.3] \xbf_i^\ssm \\ \xbf_i^\ssp \end{bmatrix}, \quad
  \Wbf_{i, j} &= \begin{bmatrix}[1.3]
    \Wbf_{i, j}^{\ssm\ssm} & \Wbf_{i, j}^{\ssm\ssp}
    \\
    \Wbf_{i, j}^{\ssp\ssm} & \Wbf_{i, j}^{\ssp\ssp}
  \end{bmatrix} , \quad
  \Bbf_i = \begin{bmatrix}[1.3] \Bbf_i^\ssm \\ \zeros \end{bmatrix}, 
  \\
  \cbf_i &= \begin{bmatrix}[1.3] \cbf_i^\ssm \\ \cbf_i^\ssp \end{bmatrix}, \quad 
  \mbf_i = \begin{bmatrix}[1.3] \mbf_i^\ssm \\ \mbf_i^\ssp,\end{bmatrix}
\end{align}
where $\xbf_i^\ssm \in \real^{r_i}$ for all $i, j \in \until{N}$. By
convention, $\Wbf_{1, 0} = \zeros$, $\Wbf_{N, N+1} = \zeros$, and $r_1
= 0$ (so $\Bbf_1 = \zeros$ and the first subnetwork has no inhibited
part).  We assume that the hierarchical layers have sufficient
timescale separation, i.e.,
\begin{align}\label{eq:tau}
  \tau_1 \gg \tau_2 \gg \dots \gg \tau_N.
\end{align} 
Finally, let
$\epsilonb = (\epsilon_1, \dots, \epsilon_{N-1})$, with
\begin{align}\label{eq:eps}
  \epsilon_i = \tau_{i+1}/\tau_i, \qquad i = \until{N-1}.
\end{align} 
Next, we first develop the main concepts and results for
the case of bilayer networks (Section~\ref{sec:bilayer}) and then
extend them to the setup with $N$ layers
(Section~\ref{sec:multilayer}).
 
\section{Selective Recruitment in Bilayer Networks}\label{sec:bilayer}

In this section we tackle the analysis of simultaneous selective
inhibition and top-down recruitment in a two-layer network. We
consider the same dynamics as in~\eqref{eq:dyn-multi} for the
lower-level subnetwork $\Nc_2$, but temporarily allow the dynamics of
$\Nc_1$ to be arbitrary.  This setup allows us to study the key
ingredients of selective recruitment without the extra complications
that arise from the multilayer interconnections of linear-threshold
subnetworks and is the basis for our later developments.  Further, by
keeping the higher-level dynamics arbitrary, the results presented
here are also of independent interest beyond HSR, as they allow for a
broader range of external inputs $\dbf(t)$ than those generated by
linear-threshold dynamics. This can be of interest in, for example,
direct brain stimulation applications where $\xbf_1(t)$ is generated
and applied by a computer in order to control the activity $\xbf_2(t)$
of certain areas of the brain. In this view, appropriate stimulation
signals $\xbf_1(t)$ may be considered as an augmentation of the
natural hierarchy of the brain if they vary slow enough to
satisfy~\eqref{eq:tau}.  Section~\ref{sec:multilayer} builds on the
insights obtained here to generalize this framework to the multilayer
case described in Section~\ref{sec:prob-state}.

For any $\Wbf \in \real^{n \times n}$, define $h: \real^n
\rightrightarrows \realnonneg^n$ by
\begin{align}\label{eq:h}
  h(\dbf) = h_{\Wbf, \mbf}(\dbf) \triangleq \setdef{\xbf \in
    \realnonneg^n}{\xbf = [\Wbf \xbf + \dbf]_\zeros^\mbf},
\end{align}
which maps any constant input~$\dbf \in \real^n$ to the corresponding
set of the equilibria of~\eqref{eq:dyn}.  Due to the switched-affine
form of the dynamics, $h$ has the piecewise-affine form
\begin{align}\label{eq:h-pa}
  h(\dbf) &= \setdef{\Fbf_\sigmab \dbf + \fbf_\sigmab}{\Gbf_\sigmab
    \dbf + \gbf_\sigmab \ge \zeros, \, \sigmab \in \zls^n},
  \\
  \notag \Fbf_\sigmab &= (\Ibf - \Sigmab^\ell \Wbf)^{-1} \Sigmab^\ell,
  \ \ \fbf_\sigmab = (\Ibf - \Sigmab^\ell \Wbf)^{-1} \Sigmab^\s \mbf,
  \\
  \notag \Gbf_\sigmab &= \begin{bmatrix} \Sigmab^\ell + \Sigmab^\s -
    \Ibf & \Sigmab^\ell & -\Sigmab^\ell & \Sigmab^\s \end{bmatrix}^T
  \Fbf_\sigmab,
  \\
  \notag \gbf_\sigmab &= \begin{bmatrix} \fbf_\sigmab^T(\Sigmab^\ell
    \!\!+\!\! \Sigmab^\s \!\!-\!\! \Ibf) & \!\!\fbf_\sigmab^T
    \Sigmab^\ell & \!\!(\mbf \!\!-\!\! \fbf_\sigmab)^T \Sigmab^\ell &
    \!\!(\fbf_\sigmab \!\!-\!\! \mbf)^T \Sigmab^\s \end{bmatrix}^T.
\end{align}
The existence and uniqueness of equilibria of~\eqref{eq:dyn} precisely
corresponds to $h$ being single-valued on $\real^n$, in which case we
let $h: \real^n \to \realnonneg^n$ be an ordinary function.  For our
subsequent analysis we need $h$ to be Lipschitz, as stated next. The
proof of this result is a special case of Lemma~\ref{lem:h-lip-gen}
and thus omitted.

\begin{lemma}\longthmtitle{Lipschitzness of $h$}\label{lem:h-lip}
  Let $h$ be as in~\eqref{eq:h} and single-valued%
  \footnote{It is possible to show, using the same proof technique,
    that $h$ is Lipschitz in the Hausdorff metric even when it is
    multiple-valued (recall that the Hausdorff distance between two
    sets $S_1, S_2 \in \real^n$ is defined as $\max\{\sup_{\abf \in
      S_1} \inf_{\bbf \in S_2} \|\abf - \bbf\|, \sup_{\bbf \in S_2}
    \inf_{\abf \in S_1} \|\abf - \bbf\|\}$).}%
  on $\real^n$. Then, it is globally Lipschitz on $\real^n$.
\end{lemma}

The main result of this section is as follows. 

\begin{theorem}\longthmtitle{Selective recruitment in bilayer %
     networks}\label{thm:sp-inhib}
   Consider the multilayer dynamics~\eqref{eq:dyn-multi} where $N =
   2$, $\Wbf_{2, 1}^{\ssm\ssp} = \zeros$, and $\cbf_2^\ssm = \zeros$
   but $\xbf_1(t)$ is generated by the dynamics
   \begin{align}\label{eq:dtilde}
     \tau_1 \dot \xbf_1(t) = \gamma(\xbf_1(t), \xbf_2(t), t).
   \end{align} 
   Let $h_2^\ssp = h_{\Wbf_{2, 2}^{\ssp\ssp}, \mbf_2^\ssp}$ as
   in~\eqref{eq:h}. If
   \begin{enumerate}
   \item $\gamma$ is measurable in $t$, locally bounded, and locally
     Lipschitz in $(\xbf_1, \xbf_2)$ uniformly in $t$;
   \item \eqref{eq:dtilde} has bounded solutions uniformly in
     $\xbf_2(t)$;
   \item $p_2 \ge r_2$;
   \item $\Wbf_{2, 2}^{\ssp\ssp}$ is such that $\tau \dot \xbf_2^\ssp
     = -\xbf_2^\ssp + [\Wbf_{2, 2}^{\ssp\ssp} \xbf_2^\ssp +
     \dbf_1^\ssp]_\zeros^{\mbf_2^\ssp}$ is GES towards a unique
     equilibrium for any constant $\dbf_1^\ssp$;
  \end{enumerate}
  then there exists $\Kbf_2 \in \real^{p_2 \times n_2}$ such that by
  using the feedback control $\ubf_2(t) = \Kbf_2 \xbf_2(t)$, one has
  \begin{align}\label{eq:tik}
    \lim\limits_{\epsilon_1 \to 0} \ \sup\limits_{t \in [\underline t
      , \bar t]} \Big\|\xbf_2(t) - \Big(\zeros_{r_2},
    h_2^\ssp\big(\Wbf_{21}^{\ssp\ssp} \xbf_1(t) +
    \cbf_2^\ssp\big)\Big)\Big\| = 0, 
  \end{align}
  for any $0 < \underline t < \bar t < \infty$, with $\epsilon_1$
  given in~\eqref{eq:eps}. Further, if the dynamics of $\xbf_2$ is
  monotonically bounded%
  \footnote{See~\cite[Def V.1]{EN-JC:18-tacI}.}, there also exists a
  feedforward control $\ubf_2(t) \equiv \bar \ubf_2$ such
  that~\eqref{eq:tik} holds for any $0 < \underline t < \bar t <
  \infty$ and arbitrary $\Wbf_{2, 1}^{\ssm\ssp}$ and~$\cbf_2^\ssm$.
\end{theorem}
\begin{proof}
  First we prove the result for feedback control. By~\emph{(iii)},
  there exists $\Kbf_2 \in \real^{p_2 \times n_2}$ almost always
  (i.e., for almost all $(\Wbf_{2, 2}^{\ssm\ssm}, \Wbf_{2,
    2}^{\ssm\ssp}, \Bbf_2^\ssm)$) such that
  \begin{align}\label{eq:w_bk}
    \Wbf_{2, 2} + \Bbf_2 \Kbf_2 = \begin{bmatrix} \zeros & \zeros \\ \Wbf_{2, 2}^{\ssp\ssm} &
      \Wbf_{2, 2}^{\ssp\ssp} \end{bmatrix}.
  \end{align}
  Further, by~\cite[Thm IV.7(ii) \& Thm V.3(ii)]{EN-JC:18-tacI}, all
  the principal submatrices of $-\Ibf + (\Wbf_{2, 2} + \Bbf_2 \Kbf_2)$
  are Hurwitz. Therefore, by~\cite[Thm IV.3 \&
  Assump.~1]{EN-JC:18-tacI}, $h_2^\ssp$ is singleton-valued almost
  always (i.e., for almost all $\Wbf_{2, 2}$). Thus, the
  $\xbf_2$-dynamics is
  \begin{align}\label{eq:dyn-d-sim}
    \tau_2 \dot \xbf_2^\ssm &= -\xbf_2^\ssm,
    \\
    \notag \tau_2 \dot \xbf_2^\ssp &= -\xbf_2^\ssp + [\Wbf_{2, 2}^{\ssp\ssm}
    \xbf_2^\ssm + \Wbf_{2, 2}^{\ssp\ssp} \xbf_2^\ssp + \Wbf_{2,
      1}^{\ssp\ssp} \xbf_1 + \cbf_2^\ssp]_\zeros^{\mbf_2^\ssp}, 
  \end{align}
  and has a unique equilibrium for any \emph{fixed} $\xbf_1$.
  Assumption~\emph{(iv)} and~\cite[Lemma~A.2]{EN-JC:18-tacI} then
  ensure that~\eqref{eq:dyn-d-sim} is GES relative to $(\zeros_{r_2},
  h_2^\ssp(\Wbf_{21}^{\ssp\ssp} \xbf_1(t) + \cbf_2^\ssp))$ for any
  fixed~$\xbf_1$.
  
  Based on assumption~\emph{(ii)}, let $D \subset \real^n$ be a
  compact set that contains the trajectory of the reduced-order model
  $\tau_1 \dot \xbf_1 = \gamma(\xbf_1, (\zeros_{r_2},
  h_2^\ssp(\Wbf_{21}^{\ssp\ssp} \xbf_1(t) + \cbf_2^\ssp)), t)$.  By
  assumption~\emph{(i)}, $\gamma$ is Lipschitz in $(\xbf_1, \xbf_2)$
  on compacts uniformly in $t$. Let $L_\gamma$ be its associated
  Lipschitz constant on $D \times \{\zeros_{r_2}\} \times
  h_2^\ssp(\Wbf_{2, 1}^{\ssp\ssp} D +
  \cbf_2^\ssp)$. Using~\eqref{eq:h-pa} and Lemma~\ref{lem:h-lip}, for
  all $\xbf_1, \hat \xbf_1 \in D$,
  \begin{align*}
    &\|\gamma(\xbf_1, h_2^\ssp(\Wbf_{2, 1}^{\ssp\ssp} \xbf_1 +
    \cbf_2^\ssp), t) - \gamma(\hat \xbf_1, h_2^\ssp(\Wbf_{2,
      1}^{\ssp\ssp} \hat \xbf_1 + \cbf_2^\ssp), t)\|
    \\
    &\le L_\gamma \|(\xbf_1 - \hat \xbf_1, h_2^\ssp(\Wbf_{2,
      1}^{\ssp\ssp} \xbf_1 + \cbf_2^\ssp) - h_2^\ssp(\Wbf_{2,
      1}^{\ssp\ssp} \hat \xbf_1 + \cbf_2^\ssp))\|
    \\
    &\le L_\gamma(\|\xbf_1 - \hat \xbf_1\| \!+\! \|h_2^\ssp(\Wbf_{2,
      1}^{\ssp\ssp} \xbf_1 + \cbf_2^\ssp) \!-\! h_2^\ssp(\Wbf_{2,
      1}^{\ssp\ssp} \hat \xbf_1 + \cbf_2^\ssp)\|)
    \\
    &\le L_\gamma(1 + L_h \|\Wbf_{2, 1}^{\ssp\ssp}\|) \|\xbf_1 - \hat
    \xbf_1\|,
  \end{align*}
  so $\gamma(\cdot, h_2^\ssp(\Wbf_{2, 1}^{\ssp\ssp} \cdot +
  \cbf_2^\ssp), t): \real^{n_1} \to \real^{n_1}$ is $L_\gamma(1 + L_h
  \|\Wbf_{2, 1}^{\ssp\ssp}\|)$-Lipschitz on~$D$. Using this,
  Lemma~\ref{lem:h-lip-gen} again, and the change of variables $t'
  \triangleq t/\tau_1$, the claim follows from~\cite[Prop 1]{VV:97}.%
  \footnote{\cite[Prop 1]{VV:97} is applicable to singularly perturbed
    differential inclusions and thus technically involved, but for
    non-smooth ODEs such as~\eqref{eq:dyn-multi}, its assumptions can
    be simplified to: 1. Lipschitzness of dynamics uniformly in $t$,
    2. Existence, uniqueness, and Lipschitzness of the equilibrium map
    of fast dynamics, 3. Lipschitzness and boundedness of the
    reduced-order model, 4. asymptotic stability of the fast dynamics
    uniformly in $t$ and the slow variable, and 5. global attractivity
    of fast dynamics for any fixed slow variable.}
 
  Next, we prove the result for constant feedforward control $\ubf_2(t)
  \equiv \bar \ubf_2$. Based on assumption~\emph{(ii)}, let $\bar
  \xbf_1 \in \realpos^{n_1}$ be the bound on the trajectories
  of~\eqref{eq:dtilde} and $\bar \ubf_2$ be a solution of
  \begin{align*}
    \Bbf_2^\ssm \bar \ubf_2 = -[[\Wbf_{2, 2}^{\ssm\ssm} \ \Wbf_{2, 2}^{\ssm\ssp}]]_\zeros^\infty
    \nub(\bar \xbf_1) - [\Wbf_{2, 1}^{\ssm\ssp}]_\zeros^\infty \bar \xbf_1 + [\cbf_2^\ssm]_\zeros^\infty, %
  \end{align*}
  where $\nub$ comes from the monotone boundedness of the dynamics of
  $\xbf_2$. This solution almost always exists by
  assumption~\emph{(ii)}. Then, the dynamics of $\xbf_2$ simplifies
  to~\eqref{eq:dyn-d-sim}, and~\cite[Lemma~A.2]{EN-JC:18-tacI}
  guarantees that it is GES relative to $(\zeros_{r_2},
  h_2^\ssp(\Wbf_{2, 1}^{\ssp\ssp} \xbf_1 + \cbf_2^\ssp))$ for any
  \emph{fixed} $\xbf_1$. The claim then follows, similar to the
  feedback case, from~\cite[Prop 1]{VV:97}.
\end{proof}

\begin{remark}\longthmtitle{Validity of the assumptions of
    Theorem~\ref{thm:sp-inhib}.}\label{re:assumptions-reasonable}
  {\rm Assumption~\emph{(i)} is merely technical and satisfied by all
    well-known models of neuronal rate dynamics, including the
    linear-threshold model itself.  Likewise, assumption~\emph{(ii)}
    is always satisfied in reality, as the firing rates of all
    biological neuronal networks are bounded by the inverse of the
    refractory period of their neurons.  In theory, the verification
    of this assumption depends clearly on $\gamma$. If a
    linear-threshold model is used, we can instead use
    Theorem~\ref{thm:multi} and relax assumption~\emph{(ii)} to a less
    restrictive one (assumption~\emph{(i)} of
    Theorem~\ref{thm:multi}), which can in turn be verified using the
    sufficient condition in Theorem~\ref{thm:multi-EUE-GES}.
    Assumption~\emph{(iii)} requires the existence of at least as many
    inhibitory control channels as the number of nodes in $\Nc_2$ that
    are to be inhibited.  Indeed, selective inhibition is still
    possible without this assumption (cf. Theorem~\ref{thm:multi}),
    but may require excessive inhibitory resources. The most critical
    requirement is assumption~\emph{(iv)}, but is not only sufficient
    but also necessary for inhibitory stabilization (cf.~\cite[Thm
    IV.8]{EN-JC:18-tacI} for conditions on $\Wbf_{2, 2}^{\ssp\ssp}$
    that ensure this assumption as well as~\cite[Thm V.2 \&
    V.3]{EN-JC:18-tacI} for its necessity for feedforward and feedback
    inhibitory stabilization)}.  \oprocend
\end{remark}

The main conclusion of Theorem~\ref{thm:sp-inhib} is the Tikhonov-type
singular perturbation statement in~\eqref{eq:tik}. According to this
statement, the tracking error can be made arbitrarily small, i.e., for
any $\theta > 0$,
\begin{align}\label{eq:tikkk}
  |\xbf_2(t) - (\zeros_{r_2}, h_2^\ssp(\xbf_1(t)))| &\le \theta
  \ones_{n_2}, \qquad \forall t \in [\underline t, \bar t],
\end{align}
provided that $\tau_2 / \tau_1$ is sufficiently small.  As discussed
in Section~\ref{sec:intro}, this timescale separation is
characteristic of biological neuronal networks.  In general, the
smaller the time constant ratio $\tau_2/\tau_1$, the smaller the
tracking error $\theta$. As shown
in~\cite{JDM-AB-DJF-RR-JDW-XC-CP-TP-HS-DL-XW:14}, several pairs of
regions along the sensory-frontal pathways have successive time
constant ratios between $1/1.5$ and $1/2.5$, which is often (more
than) enough in simulations for~\eqref{eq:tikkk} to hold with small
enough $\theta$, as shown in Example~\ref{ex:osc} below.

An important observation regarding~\eqref{eq:tikkk} is that the
equilibrium map $h_2^\ssp$ does not have a closed-form expression, so
the reference trajectory $h_2^\ssp(\xbf_1(t))$ of the lower-level
network is only implicitly known for any given~$\xbf_1(t)$. However,
if a desired trajectory $\xib_2^\ssp(t) \in \prod_{j = r_2+1}^{n_2}
[0, m_{2, j}]$ %
for $\xbf_2^\ssp$ is known a priori, one can specify the appropriate
$\gamma$ such that $h_2^\ssp(\xbf_1(t)) = \xib_2^\ssp(t)$. To show
this, let the dynamics of $\xib_2^\ssp(t)$ be
\begin{align*}
  \tau_1 \dot \xib_2^\ssp(t) = \gamma_\xi(\xib_2^\ssp(t), t).
\end{align*}
Then, choosing $\xbf_1(t) = (\Wbf_{2, 1}^{\ssp\ssp})^{-1} \big((\Ibf -
\Wbf_{2, 2}^{\ssp\ssp}) \xib_2^\ssp(t) - \cbf_2^\ssp\big)$, %
\begin{align*}
  [\Wbf_{2, 2}^{\ssp\ssp} \xib_2^\ssp(t) + \Wbf_{2, 1}^{\ssp\ssp}
  \xbf_1(t) + \cbf_2^\ssp]_\zeros^{\mbf_2^\ssp} =
  [\xib_2^\ssp(t)]_\zeros^{\mbf_2^\ssp} = \xib_2^\ssp(t) ,
\end{align*}
which, according to~\eqref{eq:h}, implies $\xib_2^\ssp(t) =
h_2^\ssp(\xbf_1(t))$.

We use this result to illustrate the core concepts of the bilayer HSR
in a synthetic but biologically-inspired example, where a inhibitory
subnetwork generates oscillations which are selectively induced on a
lower-level excitatory subnetwork.

\begin{example}\longthmtitle{HSR of an excitatory subnetwork by
    inhibitory oscillations}\label{ex:osc}
  {\rm Consider the dynamics~\eqref{eq:dyn-multi} with $N = 2$, a
    $3$-dimensional excitatory subnetwork at the lower level, a
    $3$-dimensional inhibitory subnetwork at the higher level, and
    $\mbf_1 = \mbf_2 = \infty \ones_3$
    (Figure~\ref{fig:sim-hier}). Let
  \begin{align}\label{eq:ex-osc}
    \notag &\Wbf_{1, 1} =
    \begin{bmatrix}
      0 & -0.8 & -1.7
      \\
      -1 & 0 & -0.5
      \\
      -0.7 & -1.8 & 0
    \end{bmatrix},
    \quad \cbf_1 =
    \begin{bmatrix}
      11
      \\
      10
      \\
      10
    \end{bmatrix},
    \\
    \notag &\Wbf_{2, 2} =
    \begin{bmatrix}
      0 & 0.9 & 1.2
      \\
      0.7 & 0 & 1
      \\
      0.8 & 0.2 &
      0
    \end{bmatrix},
    \quad \Bbf_2 = 
    \begin{bmatrix}
      -1
      \\
      0
      \\
      0 
    \end{bmatrix}, \quad \cbf_2
    =
    \begin{bmatrix}
      2
      \\
      3.5
      \\
      2.5
    \end{bmatrix},
    \\
    &\Wbf_{1, 2} = \zeros, \quad \Wbf_{2, 1} = -\Ibf, \quad u_2 = 5.
  \end{align}
  This example satisfies all the assumptions of
  Theorem~\ref{thm:sp-inhib}, so we expect the actual
  $\xbf_2$-trajectory to be close to the \emph{desired}
  $\xbf_2$-trajectory $(0, h_2^\ssp(\xbf_1(t))$ provided that
  $\epsilon_1 \ll 1$. As shown in Figure~\ref{fig:sim-hier},
  $\xbf_2(t)$ and $(0, h_2^\ssp(\xbf_1(t))$ are remarkably close even
  with a mild separation of timescales, $\epsilon_1 = 0.5$.

  \begin{figure}
    \centering
    \begin{tikzpicture}
      \node(g) {\includegraphics[width=0.13\textwidth]{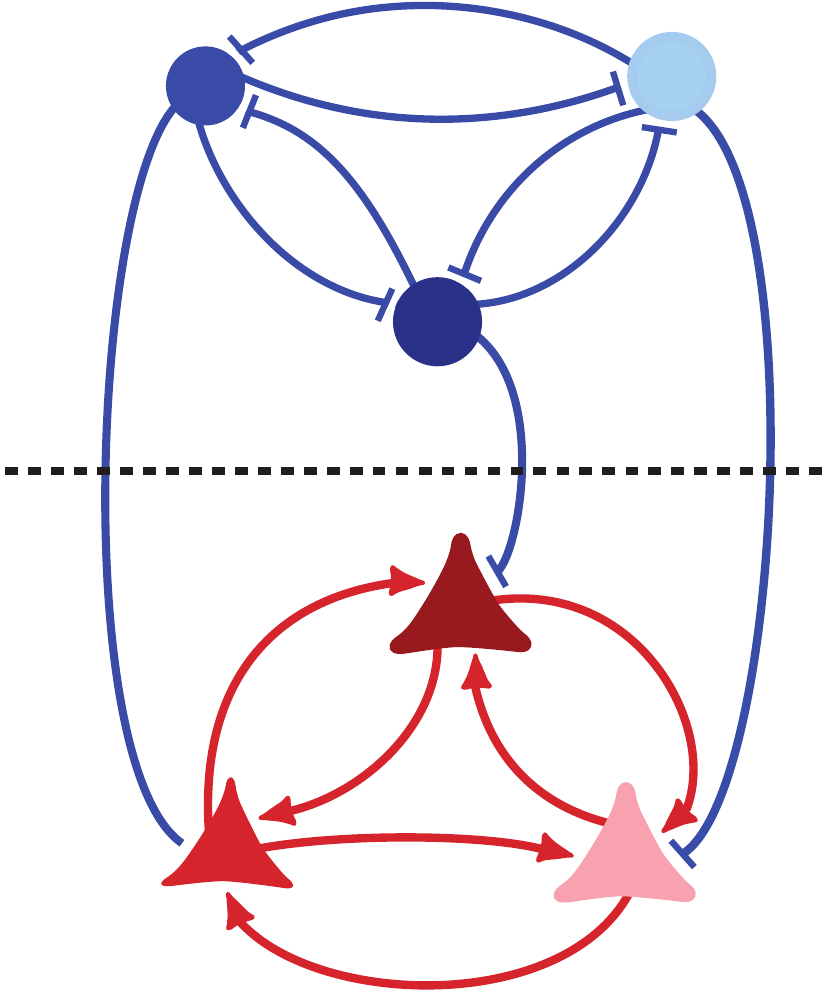}};
      \node[left of=g, xshift=-80pt, yshift=-8pt] (gt) {\includegraphics[width=0.24\textwidth]{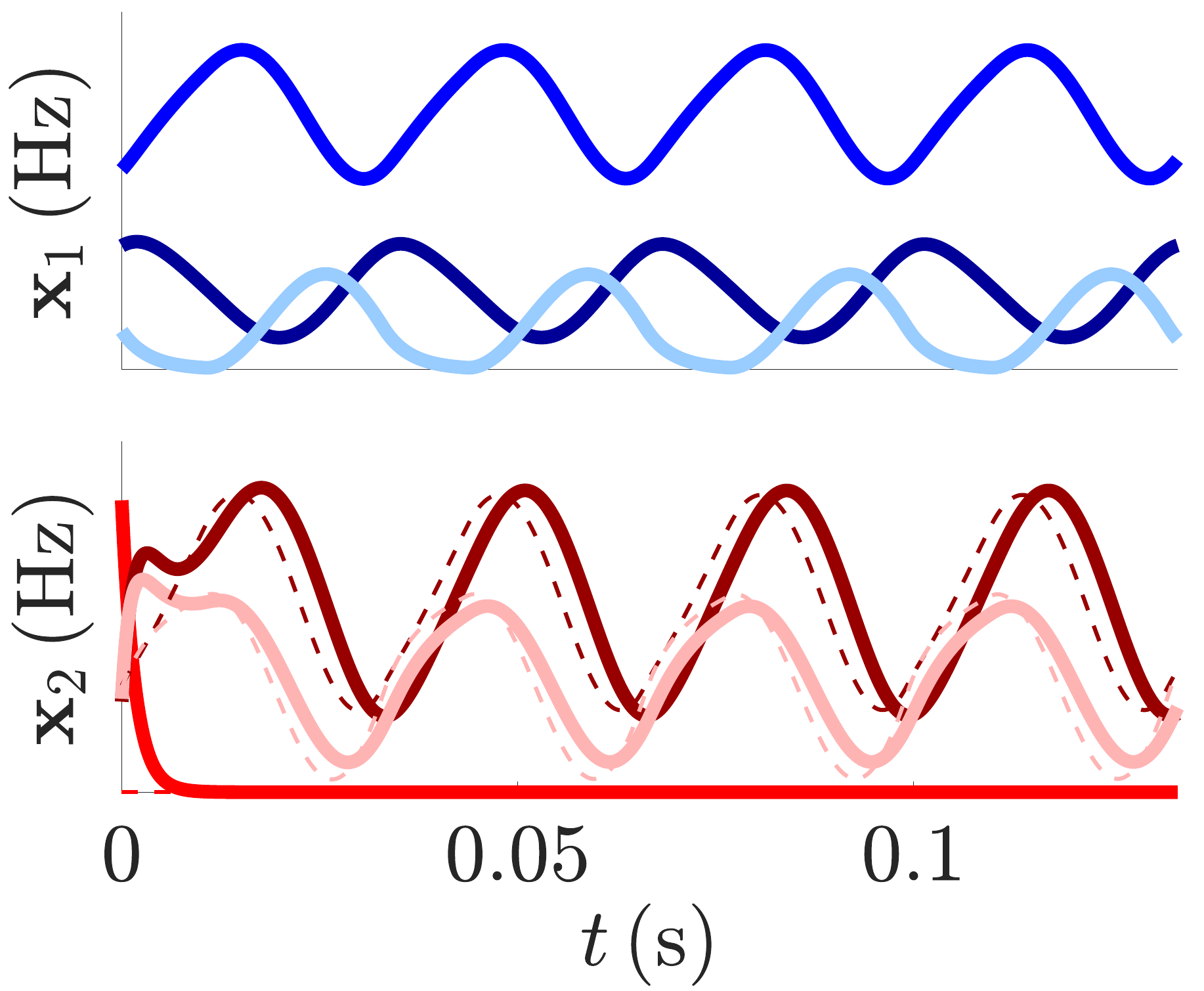}}; 
      \node[above right of=g, xshift=35pt, yshift=3pt, rotate=0, scale=0.8] (i) {\parbox{40pt}{\centering Inhibitory \\ $\tau_1$}};
      \node[below right of=g, xshift=35pt, yshift=2pt, rotate=0, scale=0.8] (e) {\parbox{45pt}{\centering Excitatory \\ $\tau_2 = 0.5 \tau_1$}};
    \end{tikzpicture}
    \caption{The network structure (right) and trajectories (left) of
      the two-timescale network in~\eqref{eq:ex-osc} for $\tau_1 =
      3.3^\text{ms}$. The red pyramids and blue circles depict
      excitatory and inhibitory nodes, resp., and the trajectory
      colors on the left correspond to node colors on the right. The
      dashed lines show the desired reference trajectories $\big(0,
      h_2^\protect\ssp(\Wbf_{2, 1}^{\protect\ssp\protect\ssp}
      \xbf_1(t) + \cbf_2^\protect\ssp)\big)$.} %
    \label{fig:sim-hier}
    \vspace{-1.5ex}
  \end{figure}
  
  This example further illustrates the complementary roles of
  selective inhibition and selective recruitment. The complete
  $\xbf_2$-subsystem is unstable by itself, but any two-dimensional
  subnetwork of it is stable. Therefore, $\Nc_1$ can selectively
  inhibit any single node of $\Nc_2$ while simultaneously recruiting
  (e.g., by inducing oscillations in) the remaining two.  Thus, as
  suggested earlier in~\cite[Rem V.7]{EN-JC:18-tacI}, different
  ``tasks'' can be accomplished at different times by varying the
  selectively recruited subnetwork over time.  Generalizing this to
  more complex networks allows for more flexible selective recruitment
  schemes of larger neuronal subnetworks, as observed in nature.}
\oprocend
\end{example}

\begin{remark}\longthmtitle{Biological relevance of
    Example~\ref{ex:osc}} {\rm In addition to providing a simple
    illustration of the HSR framework developed here,
    Example~\ref{ex:osc} has interesting similarities with well-known
    aspects of selective attention in the brain.  Extensive studies
    have demonstrated a robust correlation between oscillatory
    activity, particularly in the gamma range ($\sim
    30-100^\text{Hz}$), and selective
    attention~\cite{PF-JHR-AER-RD:01,AS-MP-WL-NB:04,SR-EN-SSH-AS-NEC:08,NK-MB-ESM-AS:12}.
    Furthermore, gamma oscillations in the cortex are shown to be
    primarily generated by networks of inhibitory neurons, which then
    recruit the excitatory populations
    (see~\cite{JAC-MC-KM-UK-FZ-KD-LT-CIM:09} and the references
    therein), as captured by the network structure of
    Figure~\ref{fig:sim-hier}. Interestingly, the oscillations
    generated by the higher-level inhibitory subnetwork fall within
    the gamma band by setting $\tau_1 \sim 3^\text{ms}$ which lies
    within the decay time constant range of $\text{GABA}_\text{A}$
    inhibitory receptors%
    \footnote{See, e.g., the Neurotransmitter Time Constants database
      of the CNRGlab at the University of Waterloo,
      \url{http://compneuro.uwaterloo.ca/research/constants-constraints/neurotransmitter-time-constants-pscs.html}.}
    (the major type of inhibitory synapse in the central nervous
    system).  \oprocend}
\end{remark}

\section{Selective \hspace*{-1pt}Recruitment \hspace*{-1pt}in
  \hspace*{-1pt}Multilayer \hspace*{-1pt}Networks}\label{sec:multilayer}

We tackle here the problem of Section~\ref{sec:prob-state} in its
general form and consider an $N$-layer hierarchical structure of
subnetworks with linear-threshold dynamics.
Given~\eqref{eq:dyn-multi}, let
\begin{align*}
 &h_i^\ssp:
    \cbf_i^\ssp \rightrightarrows \setdef{\xbf_i^\ssp}{\xbf_i^\ssp =
      [\Wbf_{i, i+1}^{\ssp\ssp} h_{i+1}^\ssp(\Wbf_{i+1, i}^{\ssp\ssp}
      \xbf_i^\ssp + \cbf_{i+1}^\ssp)
      \\
      &\hspace{92pt}+ \Wbf_{i, i}^{\ssp\ssp} \xbf_i^\ssp +
      \cbf_i^\ssp]_\zeros^{\mbf_i^\ssp}}, \ i \!=\! 2, \dots, N\!-\!1,
\end{align*}
with $h_N^\ssp = h_{\Wbf_{N, N}^{\ssp\ssp}, \mbf_N^\ssp}$, be the
recursive definition of the (set-valued) equilibrium maps of the
task-relevant parts of the layers $\{2, \dots, N\}$.  These maps play
a central role in the multiple-timescale dynamics
of~\eqref{eq:dyn-multi}. Therefore, we begin by characterizing their
piecewise-affine nature. The proof of the following result can be
found in~\ref{app:pf}.

\begin{lemma}\longthmtitle{Piecewise affinity of equilibrium maps is
    preserved along layers of hierarchical linear-threshold
    network}\label{lem:affinity}
  Let $h:\real^n \to \real^n$ be a piecewise affine function,
  \begin{align*}
    h(\cbf) = \Fbf_\lambda \cbf + \fbf_\lambda, \qquad &\forall \cbf
    \in \Psi_\lambda \triangleq \setdef{\cbf}{\Gbf_\lambda \cbf +
      \gbf_\lambda \ge \zeros},
    \\
    &\forall \lambda \in \Lambda,
  \end{align*}
  where $\Lambda$ is a finite index set and $\bigcup_{\lambda \in
    \Lambda} \Psi_\lambda = \real^n$. Given matrices $\Wbf_\ell, \ell
  = 1, 2, 3$ and a vector $\bar \cbf$, assume %
  \begin{align}\label{eq:gen-eq}
    \xbf = [\Wbf_1 \xbf + \Wbf_2 h(\Wbf_3 \xbf + \bar \cbf) + \cbf']_\zeros^\mbf,
  \end{align}
  is known to have a unique solution $\xbf \in \real^{n'}$ for all
  $\cbf' \in \real^{n'}$ and let $h'(\cbf')$ be this unique
  solution. Then, there exists a finite index set $\Lambda'$ and
  $\{(\Fbf'_{\lambda'}, \fbf'_{\lambda'}, \Gbf'_{\lambda'},
  \gbf'_{\lambda'})\}_{\lambda' \in \Lambda'}$ such that
  \begin{align}\label{eq:h'}
    \notag    h'(\cbf') = \Fbf'_{\lambda'} \cbf' + \fbf'_{\lambda'}, \quad
    &\forall \cbf' \in \Psi'_{\lambda'} \triangleq
    \setdef{\cbf'}{\Gbf'_{\lambda'} \cbf' + \gbf'_{\lambda'} \ge
      \zeros},
    \\
    &\forall \lambda' \in \Lambda',
  \end{align}
  and $\bigcup_{\lambda' \in \Lambda'} \Psi'_{\lambda'} = \real^{n'}$. \oprocend
\end{lemma}

A special case of Lemma~\ref{lem:affinity} is when $\Wbf_2 = \zeros$,
where $h'$ becomes, like $h_N^\ssp$, the standard equilibrium
map~\eqref{eq:h}. %
Next, we characterize the global Lipschitzness property of the
equilibrium maps. The proof is in~\ref{app:pf}.

\begin{lemma}\longthmtitle{Piecewise affine equilibrium maps are
    globally Lipschitz}\label{lem:h-lip-gen}
  Let $h:\real^n \to \real^n$ be a piecewise affine function of the
  form
  \begin{align*}
    h(\cbf) = \Fbf_\lambda \cbf + \fbf_\lambda, \qquad &\forall \cbf
    \in \Psi_\lambda \triangleq \setdef{\cbf}{\Gbf_\lambda \cbf +
      \gbf_\lambda \ge \zeros},
    \\
    &\forall \lambda \in \Lambda,
  \end{align*}
  where $\Lambda$ is a finite index set and $\bigcup_{\lambda \in
    \Lambda} \Psi_\lambda = \real^n$. Then, $h$ is globally Lipschitz. \oprocend
\end{lemma}

We are now ready to generalize Theorem~\ref{thm:sp-inhib} to an
$N$-layer architecture while at the same time relaxing several of its
simplifying assumptions in favor of generality.

\begin{theorem}\longthmtitle{Selective recruitment in multilayer %
    networks}\label{thm:multi}
  Consider the dynamics~\eqref{eq:dyn-multi}. If
  \begin{enumerate}
  \item The reduced-order model (ROM)
    \begin{align*}
      \hspace{-20pt} \tau_1 \dot {\bar \xbf}_1^\ssp = -\bar
      \xbf_1^\ssp + [\Wbf_{1, 1}^{\ssp\ssp} \bar \xbf_1^\ssp +
      \Wbf_{1, 2}^{\ssp\ssp} h_2^\ssp(\Wbf_{2, 1}^{\ssp\ssp} \bar
      \xbf_1^\ssp + \cbf_2^\ssp) +
      \cbf_1^\ssp]_\zeros^{\mbf_1^\ssp},
    \end{align*}
    of the first subnetwork has bounded solutions (recall $\xbf_1
    \equiv \xbf_1^\ssp$ since $r_1 = 0$);
  \item For all $i = 2, \dots, N$, 
    \begin{align*}
      \hspace{-20pt} \tau_i \dot \xbf_i^\ssp(t) = &-\xbf_i^\ssp(t) +
      [\Wbf_{i, i}^{\ssp\ssp} \xbf_i^\ssp(t)
      \\
      &+ \Wbf_{i, i+1}^{\ssp\ssp} h_{i+1}^\ssp(\Wbf_{i+1,
        i}^{\ssp\ssp} \xbf_i^\ssp(t) + \cbf_{i+1}^\ssp) +
      \cbf_i^\ssp]_\zeros^{\mbf_i^\ssp},
    \end{align*}
    is GES towards a unique equilibrium for any $\cbf_{i+1}^\ssp$
    and~$\cbf_i^\ssp$;
  \end{enumerate}
  then there exists $\Kbf_i \in \real^{p_i \times n_i}$ and $\bar
  \ubf_i: \realnonneg \to \realnonneg^{p_i}, i \in \{2, \dots, N\}$
  such that using the feedback-feedforward control
  \begin{align}\label{eq:u-gen}
    \ubf_i(t) = \Kbf_i \xbf_i(t) + \bar \ubf_i(t), \qquad i \in \{2,
    \dots, N\},
  \end{align}
  we have, for any $0 < \underline t < \bar t < \infty$,
  \begin{subequations}\label{eq:multi-tik}
    \begin{align}
      \lim_{\epsilonb \to \zeros} \sup_{t \in [\underline t, \bar t]}
      \|\xbf_i^\ssm(t)\| = \zeros, \qquad \forall i \in \{2, \dots, N\},
    \end{align}
    and
    \begin{align}
      &\!\!\!\!\lim_{\epsilonb \to \zeros} \sup_{t \in [0, \bar t]}
      \|\xbf_1^\ssp(t) - \bar \xbf_1^\ssp(t)\| = 0,
      \\
      &\!\!\!\!\lim_{\epsilonb \to \zeros} \sup_{t \in [\underline t, \bar t]}
      \|\xbf_2^\ssp(t) - h_2^\ssp(\Wbf_{2, 1}^{\ssp\ssp} \xbf_1^\ssp(t) + \cbf_2^\ssp)\|
      = 0,
      \\
      \notag &\!\!\!\!\qquad \vdots
      \\
      &\!\!\!\lim_{\epsilonb \to \zeros} \sup_{t \in [\underline t, \bar t]}
      \!\!\|\xbf_N^\ssp(t) \!-\! h_N^\ssp(\Wbf_{N, N-1}^{\ssp\ssp} \xbf_{N-1}^\ssp(t) \!+\!
      \cbf_N^\ssp)\| \!=\! 0.
    \end{align}
  \end{subequations}
\end{theorem}
\begin{proof}
  For any $2 \times 2$ block-partitioned matrix $\Wbf$, we introduce
  the convenient notation $\Wbf^{\ell,\ssmp} \triangleq [\Wbf^{\ell\ssm} \
  \Wbf^{\ell\ssp}]$
  and $\Wbf^{\ssmp,\ell} \triangleq [(\Wbf^{\ssm\ell})^T \
  (\Wbf^{\ssp\ell})^T]^T$ for $\ell = \ssm, \ssp$. Further, for any $i
  \in \{2, \dots, N\}$, let $\xbf_{1:i} = [\xbf_1^T \ \dots \
  \xbf_i^T]^T$.
  To begin with, let $\Kbf_N$ and $\bar \ubf_N$ be such that
  \begin{subequations}\label{eq:inhib-ineqs}
    \begin{align}
      \Bbf_N^\ssm \Kbf_N &\le -\Wbf_{N,
        N}^{\ssm,\ssmp}, \label{eq:inhib-ineqs1}
      \\
      \bar \ubf_N(t) &\le -\Wbf_{N, N-1}^{\ssm,\ssmp} \xbf_{N-1}(t) -
      \cbf_N^\ssm, \qquad \forall t, \label{eq:inhib-ineqs2}
    \end{align}
  \end{subequations}
  Note that, if $p_N \ge r_N$, then~\eqref{eq:inhib-ineqs1} can be
  satisfied with equality. Otherwise, \eqref{eq:inhib-ineqs1} can
  still be satisfied since all the rows of $\Bbf_N^\ssm$ are nonzero,
  but may require excessive amounts of inhibition. Also, notice that
  $\bar \ubf_N$ is set by the subnetwork $N - 1$, which has access to
  $\xbf_{N-1}(t)$ and can thus fulfill~\eqref{eq:inhib-ineqs2}. As a
  result, the nodes in $\xbf_N^\ssm$ are fully inhibited and evolve
  according to
  $\tau_N \dot \xbf_N^\ssm = -\xbf_N^\ssm$.
  The overall dynamics become
  \begin{align*}
    \tau_1 \dot \xbf_1
      &\!=\! -\xbf_1 + [\Wbf_{1, 1} \xbf_1 + \Wbf_{1, 2} \xbf_2 +
      \cbf_1]_\zeros^{\mbf_1},
      \\
      &\ \ \vdots
      \\
      \tau_{N-1} \dot \xbf_{N-1} &\!=\! -\xbf_{N-1} + [\Wbf_{N-1, N-1}
      \xbf_{N-1} \!+\! \Bbf_{N-1} \ubf_{N-1}
      \\
      &\quad \!\!\!\!\!\!\!\!\!\!\!\!+\! \Wbf_{N-1, N} \xbf_N \!+\!
      \Wbf_{N-1, N-2} \xbf_{N-2} \!+\!
      \cbf_{N-1}]_\zeros^{\mbf_{N-1}}\!\!,
      \\
      \epsilon_{N-1} \tau_{N-1} \dot \xbf_N^\ssm &\!=\! -\xbf_N^\ssm,
      \\
      \epsilon_{N-1} \tau_{N-1} \dot \xbf_N^\ssp &\!=\! -\xbf_N^\ssp
      \!+\! [\Wbf_{N, N}^{\ssp,\ssmp} \xbf_N \!+\! \Wbf_{N,
        N-1}^{\ssp,\ssmp} \xbf_{N-1} \!+\!
      \cbf_N^\ssp]_\zeros^{\mbf_N^\ssp}.
  \end{align*}
  Letting $\epsilon_{N-1} \to 0$, we get our first separation of
  timescales between $\xbf_N$ and $\xbf_{1:N-1}$, as follows. For any
  constant $\xbf_{N-1}$, the $\xbf_N$ dynamics is GES by assumption
  \emph{(ii)} and \cite[Lemma~A.2]{EN-JC:18-tacI}. Further, the
  equilibrium map $h_N = (\zeros_{r_N}, h_N^\ssp)$
  of the $N$'th subnetwork is globally Lipschitz by
  Lemmas~\ref{lem:affinity} and~\ref{lem:h-lip-gen}, and the entire
  vector field of network dynamics is globally Lipschitz due to the
  Lipschitzness of $[\, \cdot\, ]_0^m$. Therefore, it follows from~\cite[Prop
  1]{VV:97} that for any $0 < \underline t < \bar t < \infty$,
  \begin{align*}
    &\lim_{\epsilon_{N-1} \to 0} \sup_{t \in [\underline t, \bar t]}
    \|\xbf_N^\ssm(t)\| = 0,
    \\
    &\lim_{\epsilon_{N-1} \to 0} \sup_{t \in [\underline t, \bar t]}
    \|\xbf_N^\ssp(t) - h_N^\ssp(\Wbf_{N, N-1}^{\ssp,\ssmp} \xbf_{N-1}(t) +
    \cbf_N^\ssp)\| = 0,
    \\
    &\lim_{\epsilon_{N-1} \to 0} \sup_{t \in [0, \bar t]}
    \|\xbf_{1:N-1}(t) - \xbf_{1:N-1}^{(1)}(t)\| = 0.
  \end{align*}
  Here, $\xbf_{1:N-1}^{(1)}$ is the solution of the ``first-step ROM"
  \begin{align*}
    &\tau_1 \dot \xbf_1^{(1)} = -\xbf_1^{(1)} + [\Wbf_{1, 1} \xbf_1^{(1)}
    + \Wbf_{1, 2} \xbf_2^{(1)} + \cbf_1]_\zeros^{\mbf_1},
    \\
    &\ \ \vdots
    \\
    &\tau_{N-1} \dot \xbf_{N-1}^{(1)} = -\xbf_{N-1}^{(1)} + [\Wbf_{N-1, N-1}
    \xbf_{N-1}^{(1)}
    \\
    &\qquad \qquad + \Wbf_{N-1, N}^{\ssmp,\ssp} h_N^\ssp(\Wbf_{N, N-1}^{\ssp,\ssmp}
    \xbf_{N-1}^{(1)}(t) + \cbf_N^\ssp)
    \\
    &\qquad \qquad + \Wbf_{N-1, N-2} \xbf_{N-2}^{(1)} + \Bbf_{N-1}
    \ubf_{N-1} + \cbf_{N-1}]_\zeros^{\mbf_{N-1}},
  \end{align*}
  which results from replacing $\xbf_N$ with its equilibrium
  value. Except for technical adjustments, the remainder of the proof
  essentially follows by repeating this process $N - 2$ times. In
  particular, for $i = N-1, \dots, 2$, let $\Kbf_i$ and $\bar \ubf_i$
  be such that
{\interdisplaylinepenalty=10000
  \begin{align*}
    \Bbf_i^\ssm \Kbf_i &\le -|\Wbf_{i, i}^{\ssm,\ssmp}| - |\Wbf_{i, i+1}^{\ssm\ssp}| \bar
    \Fbf_{i+1} |\Wbf_{i+1, i}^{\ssp,\ssmp}|,
    \\
    \bar \ubf_i(t) &\le -\Wbf_{i, i-1}^{\ssm:} \xbf_{i-1}(t) - \cbf_i^\ssm,
    \qquad \forall t,
  \end{align*}
}
  where $\bar \Fbf_i \in \real^{(n_i - r_i) \times (n_i - r_i)}$ is the entry-wise
  maximal gain of the map $h_i^\ssp$ over $\real^{n_i - r_i}$
  (cf. Theorem~\ref{thm:multi-EUE-GES}). This results in the ``$(N - i)$'th-step ROM"
  \begin{align*}
    &\tau_1 \dot \xbf_1^{(N-i)} \!=\! -\xbf_1^{(N-i)} \!\!+\!
    [\Wbf_{1, 1} \xbf_1^{(N-i)} \!\!+\! \Wbf_{1, 2} \xbf_2^{(N-i)}
    \!\!+\! \cbf_1]_\zeros^{\mbf_1},
    \\
    &\qquad \vdots
    \\
    &\tau_{i-1} \dot \xbf_{i-1}^{(N-i)} = -\xbf_{i-1}^{(N-i)} +
    [\Wbf_{i-1, i-1} \xbf_{i-1}^{(N-i)}
    \\
    &\qquad \qquad \qquad \ \ + \Wbf_{i-1, i} \xbf_i^{(N-i)} +
    \Wbf_{i-1, i-2} \xbf_{i-2}^{(N-i)}
    \\
    &\qquad \qquad \qquad \ \ + \Bbf_{i-1} \ubf_{i-1} + \cbf_{i-1}]_\zeros^{\mbf_{i-1}},
    \\
    &\epsilon_{i-1} \tau_{i-1} \dot \xbf_i^{(N-i)\ssm} =
    -\xbf_i^{(N-i)\ssm},
    \\
    &\epsilon_{i-1} \tau_{i-1} \dot \xbf_i^{(N-i)\ssp} =
    -\xbf_i^{(N-i)\ssp} + [\Wbf_{i, i}^{\ssp,\ssmp} \xbf_i^{(N-i)\ssp}
    \\
    &\qquad \qquad \qquad \quad \ + \Wbf_{i, i+1}^{\ssmp,\ssp}
    h_{i+1}^\ssp(\Wbf_{i+1, i}^{\ssp,\ssmp} \xbf_i^{(N-i)}(t) \!+\!
    \cbf_{i+1}^\ssp)
    \\
    &\qquad \qquad \qquad \quad \ + \Wbf_{i, i-1}^{\ssp,\ssmp}
    \xbf_{i-1}^{(N-i)} + \cbf_i^\ssp]_\zeros^{\mbf_i^\ssp}.
  \end{align*}
  Similarly to above, invoking~\cite[Prop 1]{VV:97} then ensures that
  \begin{align*}
    &\lim_{\epsilonb \to \zeros} \sup_{t \in [\underline t, \bar t]}
    \|\xbf_i^{(N-i)\ssm}(t)\| = 0,
    \\
    &\lim_{\epsilonb \to \zeros} \sup_{t \in [\underline t, \bar t]}
    \|\xbf_i^{(N-i)\ssp}(t) \!-\! h_i^\ssp(\Wbf_{i, i-1}^{\ssp,\ssmp}
    \xbf_{i-1}^{(N-i)}(t) \!+\!  \cbf_i^\ssp)\| = 0,
    \\
    &\lim_{\epsilonb \to 0} \sup_{t \in [0, \bar t]}
    \|\xbf_{1:i-1}^{(N-i)}(t) - \xbf_{1:i-1}^{(N-i+1)}(t)\| = 0.
  \end{align*}
  Note that, after every invocation of~\cite[Prop 1]{VV:97}, the
  super-index inside the parenthesis increases by $1$, showing one
  more replacement of a fast dynamics by its equilibrium state. In
  particular, after the $(N - 1)$'th invocation of~\cite[Prop
  1]{VV:97}, we reach $\xbf_1^{(N-1)\ssp}$, which is the same as $\bar
  \xbf_1^\ssp$ in the statement. Together, these results (and
  sufficiently many applications of the triangle inequality and
  Lemma~\ref{lem:h-lip-gen}) ensure~\eqref{eq:multi-tik}.
\end{proof}

An instructive difference, by design, between
Theorems~\ref{thm:sp-inhib} and~\ref{thm:multi} is the separate
treatment of feedforward and feedback inhibition in the former versus
the combination of the two in the latter. While the separate treatment
in Theorem~\ref{thm:sp-inhib} is conceptually simpler and highlights
the theoretical difference between the two inhibitory mechanisms, the
combination in Theorem~\ref{thm:multi} results in more flexibility and
less conservativeness: in pure feedforward inhibition, countering
local excitations requires monotone boundedness and a sufficiently
large $\bar \ubf$ providing inhibition under the worst-case scenario,
a goal that is achieved more efficiently using feedback.  On the other
hand, pure feedback inhibition needs to dynamically cancel local
excitations at all times and is also unable to counter the effects of
constant background excitation, limitations that are easily addressed
when combined with feedforward inhibition.

Similar to Theorem~\ref{thm:sp-inhib}
(cf. Remark~\ref{re:assumptions-reasonable}), assumption~\emph{(ii)}
of Theorem~\ref{thm:multi} is its only critical requirement which we
next relate to the joint structure of the subnetworks.

\begin{theorem}\longthmtitle{Sufficient condition for existence and
    uniqueness of equilibria and GES in multilayer linear-threshold
    networks}\label{thm:multi-EUE-GES} 
  Let $h:\real^n \to \real^n$ be a piecewise affine function of the
  form
  \begin{align}\label{eq:h-aff-form}
    \notag h(\cbf) = \Fbf_\lambda \cbf + \fbf_\lambda, \qquad &\forall \cbf
    \in \Psi_\lambda \triangleq \setdef{\cbf}{\Gbf_\lambda \cbf +
      \gbf_\lambda \ge \zeros},
    \\
    &\forall \lambda \in \Lambda,
  \end{align}
  where $\Lambda$ is a finite index set and $\bigcup_{\lambda \in
    \Lambda} \Psi_\lambda = \real^n$. Further, let $\bar \Fbf
  \triangleq \max_{\lambda \in \Lambda}|\Fbf_\lambda|$ be the matrix
  whose elements are the maximum of the corresponding elements from
  $\{|\Fbf_\lambda|\}_{\lambda \in \Lambda}$. For arbitrary matrices
  $\Wbf_\ell$, $ \ell = 1, 2, 3$, if $ \rho\big(|\Wbf_1| + |\Wbf_2|
  \bar \Fbf |\Wbf_3|\big) < 1$, then the linear-threshold dynamics
  \begin{align*}
    \tau \dot \xbf(t) = -\xbf(t) + [\Wbf_1 \xbf(t) + \Wbf_2 h(\Wbf_3
    \xbf(t) + \bar \cbf) + \cbf]_\zeros^\mbf,
  \end{align*}
  is GES towards a unique equilibrium
  for all $\bar \cbf$ and $\cbf$.
\end{theorem}
\begin{proof}
  We use the same proof technique as in~\cite[Prop. 3]{JF-KPH:96}.
  For simplicity, assume that $|\Wbf_1| + |\Wbf_2| \bar \Fbf |\Wbf_3|$
  is irreducible (i.e., the network topology induced by it is strongly
  connected)%
  \footnote{If $|\Wbf_1| + |\Wbf_2| \bar \Fbf |\Wbf_3|$ is not
    irreducible, it can be ``upper-bounded'' by the irreducible matrix
    $|\Wbf_1| + |\Wbf_2| \bar \Fbf |\Wbf_3| + \mu \ones_n \ones_n^T$,
    with $\mu > 0$ sufficiently small such that $\rho(|\Wbf_1| +
    |\Wbf_2| \bar \Fbf |\Wbf_3| + \mu \ones_n \ones_n^T) < 1$. The
    same argument can then be employed for this upper bound.}.
  Then, the left Perron-Frobenius eigenvector $\alphab$ of $|\Wbf_1| +
  |\Wbf_2| \bar \Fbf |\Wbf_3|$ has positive entries~\cite[Fact
  4.11.4]{DSB:09}, making the map $\|\cdot\|_\alphab: \vbf \to
  \|\vbf\|_\alphab \triangleq \alphab^T |\vbf|$ a norm on $\real^n$.
  Further, it can be shown, similar to the proof of
  Lemma~\ref{lem:h-lip-gen}, that for all $\cbf_1, \cbf_2 \in
  \real^n$,
    $|h(\cbf_1) - h(\cbf_2)| \le \bar \Fbf |\cbf_1 - \cbf_2|$,
  where the inequality is entrywise. Thus, for any $\xbf, \hat \xbf
  \in \real^n$, 
  \begin{align*}
      &\big\|[\Wbf_1 \xbf + \Wbf_2 h(\Wbf_3 \xbf + \wbf) + \cbf]_\zeros^\mbf
      \\
      &\qquad \qquad \qquad - [\Wbf_1 \hat \xbf + \Wbf_2 h(\Wbf_3 \hat
      \xbf + \wbf) + \cbf]_\zeros^\mbf\big\|_\alphab
      \\
      &= \alphab^T \big|[\Wbf_1 \xbf + \Wbf_2 h(\Wbf_3 \xbf + \wbf) +
      \cbf]_\zeros^\mbf
      \\
      &\qquad \qquad \qquad - [\Wbf_1 \hat \xbf + \Wbf_2 h(\Wbf_3 \hat
      \xbf + \wbf) + \cbf]_\zeros^\mbf\big|
      \\
      &\le \alphab^T \big|\Wbf_1(\xbf - \hat \xbf) + \Wbf_2(h(\Wbf_3
      \xbf + \wbf) - h(\Wbf_3 \hat \xbf + \wbf))\big|
      \\
      &\le \alphab^T \big(|\Wbf_1| + |\Wbf_2| \bar \Fbf |\Wbf_3|\big)
      |\xbf - \hat \xbf|
      \\
      &= \rho\big(|\Wbf_1| + |\Wbf_2| \bar \Fbf |\Wbf_3|\big)
      \alphab^T |\xbf - \hat \xbf|
      \\
      &= \rho\big(|\Wbf_1| + |\Wbf_2| \bar \Fbf |\Wbf_3|\big) \|\xbf -
      \hat \xbf\|_\alphab.
    \end{align*}
    This proves that $\xbf \mapsto [\Wbf_1 \xbf + \Wbf_2 h(\Wbf_3 \xbf
    + \wbf) + \cbf]_\zeros^\mbf$ is a contraction (on $\realnonneg^n$
    if $\mbf = \infty \ones_n$ or on $\prod_i [0, m_i]$ if $\mbf <
    \infty \ones_n$) and has a unique fixed point, denoted $\xbf^*$,
    by the Banach Fixed-Point Theorem~\cite[Thm 9.23]{WR:76}.
  
  To show GES, let $\xib(t) \triangleq (\xbf(t) - \xbf^*) e^t$,
  satisfying
  \begin{align}\label{eq:xi-dot-2}
    \tau \dot \xib(t) = \Mbf(t) \Wbf \xib(t),
  \end{align}
  where $\Mbf(t)$ is a diagonal matrix with diagonal entries
\begin{align*}
    M_{ii}(t) \triangleq
      \frac{\big([\Wbf_1 \xbf(t) + \Wbf_2 h(\Wbf_3 \xbf(t) + \wbf) +
        \cbf]_\zeros^\mbf - \xbf^*)_i}{\xi_i(t)} &
  \end{align*}
  if $\xi_i(t) \neq 0$ and $M_{ii}(t) \triangleq 0$ otherwise. Then
  \begin{align*}
  |\Mbf(t)| \le |\Wbf_1| + |\Wbf_2| \bar \Fbf |\Wbf_3|, \qquad \forall t \ge 0,
  \end{align*} 
  where the inequality is entry-wise. Then, by using~\cite[Lemma]{KPH-DK:87}
  (which is essentially a careful application of Gronwall-Bellman's
  Inequality~\cite[Lemma A.1]{HKK:02} to~\eqref{eq:xi-dot-2}),
  \begin{align*}
    &\|\xib(t)\|_\alphab \le \|\xib(0)\|_\alphab e^{\rho(|\Wbf_1| + |\Wbf_2| \bar \Fbf |\Wbf_3|) t}
    \\
    &\Rightarrow \|\xbf(t) \!-\! \xbf^*\|_\alphab \!\!\le \|\xbf(0) \!-\!
    \xbf^*\|_\alphab e^{-(1 - \rho(|\Wbf_1| + |\Wbf_2| \bar \Fbf |\Wbf_3|)) t},
  \end{align*}
  establishing GES by the equivalence of norms on $\real^n$.
\end{proof}

Note that Theorem~\ref{thm:multi-EUE-GES} applies to each layer
of~\eqref{eq:dyn-multi} separately. When put together,
Theorem~\ref{thm:multi}\emph{(ii)} is satisfied if
\begin{align}\label{eq:rec-cond-sim}
  \notag &\rho\big(|\Wbf_{2, 2}^{\ssp\ssp}| + |\Wbf_{2, 3}^{\ssp\ssp}|
  \bar \Fbf_3^\ssp |\Wbf_{3, 2}^{\ssp\ssp}|\big) < 1,
  \\
  \notag &\qquad \vdots
  \\
  \notag &\rho\big(|\Wbf_{N-1, N-1}^{\ssp\ssp}| + |\Wbf_{N-1,
    N}^{\ssp\ssp}| \bar \Fbf_N^\ssp |\Wbf_{N, N-1}^{\ssp\ssp}|\big) <
  1,
  \\
  &\rho\big(|\Wbf_{N, N}^{\ssp\ssp}|\big) < 1,
\end{align}
where $\bar \Fbf_i^\ssp, i = 2, \dots, N$ is the matrix described in
Theorem~\ref{thm:multi-EUE-GES} corresponding to $h_i^\ssp$, and the
affine form~\eqref{eq:h-aff-form} of $h_i^\ssp$ is computed
recursively using Lemma~\ref{lem:affinity}.  Moreover, if $\mbf^\ssp_1
= \infty \ones_{r_1}$, then $\rho\big(|\Wbf_{1, 1}^{\ssp\ssp}| +
|\Wbf_{1, 2}^{\ssp\ssp}| \bar \Fbf_2^\ssp |\Wbf_{2,
  1}^{\ssp\ssp}|\big) < 1$ serves as a sufficient condition for
Theorem~\ref{thm:multi}\emph{(i)} (which is trivial if $\mbf^\ssp_1 <
\infty \ones_{r_1}$).

\section{Case Study: Selective Listening in Rodents}

We present an application of our framework to a specific real-world
example of goal-driven selective attention using measurements of
single-neuron activity in the brain.  Beyond the conceptual
illustration of our results in Example~\ref{ex:osc} above, we argue
that the cross-validation of theoretical results with real data
performed here is a necessary step to make a credible case for
neuroscience research and significantly enhances the relevance of the
developed analysis.  We have been fortunate to have access to data
from a novel and carefully designed experimental
paradigm~\cite{CCR-MRD:14,CCR-MRD:14-crcns} that involves goal-driven
selective listening in rodents and displays the key features of
hierarchical selective recruitment noted~here.

\subsection{Description of Experiment and Data}

A long standing question in neuroscience involves our capability to
selectively listen to specific sounds in a crowded
environment~\cite{ECC:53,AWB:15}.  To understand the neuronal basis of
this phenomena, the work~\cite{CCR-MRD:14} has rats simultaneously
presented with two sounds and trains them to selectively respond to
one sound while actively suppressing the distraction from the other. %
In each trial, the animal simultaneously hears a white noise burst and
a narrow-band warble. The noise burst may come from the left or the
right while the warble may have low or high pitch, both chosen at
random. Which of the two sounds (noise burst or warble) is relevant
and which is a distraction depends on the ``rule'' of the trial: in
``localization'' (LC) and ``pitch discrimination'' (PD) trials, the
animal has to make a motor choice based on the location of the noise
burst (left/right) or the pitch of the warble (low/high), resp., to
receive a reward. Each rat performs several blocks of LC and PD trials
during each session (with each block switching randomly between the 4
possible stimulus pairs), requiring it to quickly switch its response
following the rule changes.

While the rats perform the task, spiking activity of single neurons is
recorded in two brain areas: the primary auditory cortex (A1) and the
medial prefrontal cortex (PFC). A1 is the first region in the cortex
that receives auditory information (from subcortical areas and ears),
thus forming a (relatively) low level of the hierarchy. PFC is
composed of multiple regions that form the top of the hierarchy, and
serve functions such as imagination, planning, decision-making, and
attention~\cite{JF:15}. Spike times of 211 well-isolated and
reliable neurons are recorded in $5$ rats, $112$ in PFC and $99$ in
A1, see~\cite{CCR-MRD:14-crcns}.

Using statistical analysis, it was shown in~\cite{CCR-MRD:14} that (i)
the rule of the trial and the stimulus sounds are more strongly
encoded by PFC and A1 neurons, resp., (ii) electrical disruption of
PFC significantly impairs task performance, and (iii) PFC activity
temporally precedes A1 activity. These findings are all consistent
with a model where PFC controls the activity of A1 based on the trial
rule in order to achieve GDSA. We next build on these observations to
define an appropriate network structure and rigorously analyze it
using HSR.
 
\subsection{Choice of Neuronal Populations}

To form meaningful populations among the recorded neurons, we perform
three classifications of them:
\begin{enumerate}[wide]
\item first, we classify the neurons into excitatory and
  inhibitory. The procedure for this classification is based on~the
  neuron's spike waveform: excitatory neurons have slower and
  wider spikes while inhibitory neurons have faster and narrower
  ones~\cite{RMB-DJS:02}.  Using standard k-means clustering on the
  $24$-dimensional spike waveform time-series, we identify $174$
  excitatory and $36$ inhibitory neurons~\footnote{The type of one
    neuron could not be identified with confidence and was discarded
    from further analysis.}  (Figure~\ref{fig:ei}(a)).  These clusters
  conform with spike width difference of excitatory and inhibitory
  neurons (Figure~\ref{fig:ei}(b)) and the common expectation that
  about $80\%$ of mammalian cortical neurons are~excitatory.
  
  \begin{figure}
    \subfloat[]{
      \includegraphics[width=0.5\linewidth]{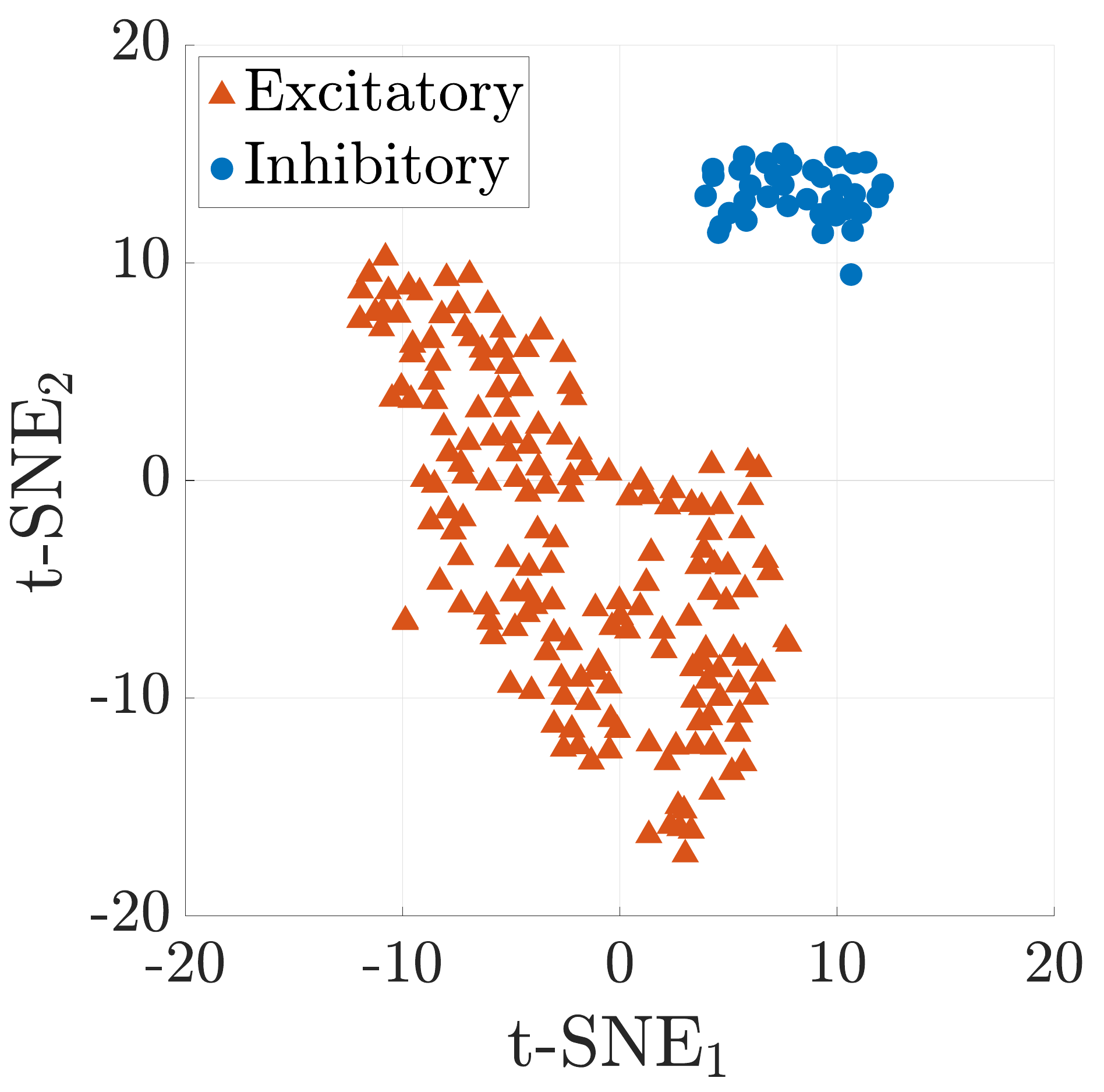}
    }
    \subfloat[]{
      \parbox{0.4\linewidth}{
        \vspace*{-124pt}
        \includegraphics[width=0.95\linewidth]{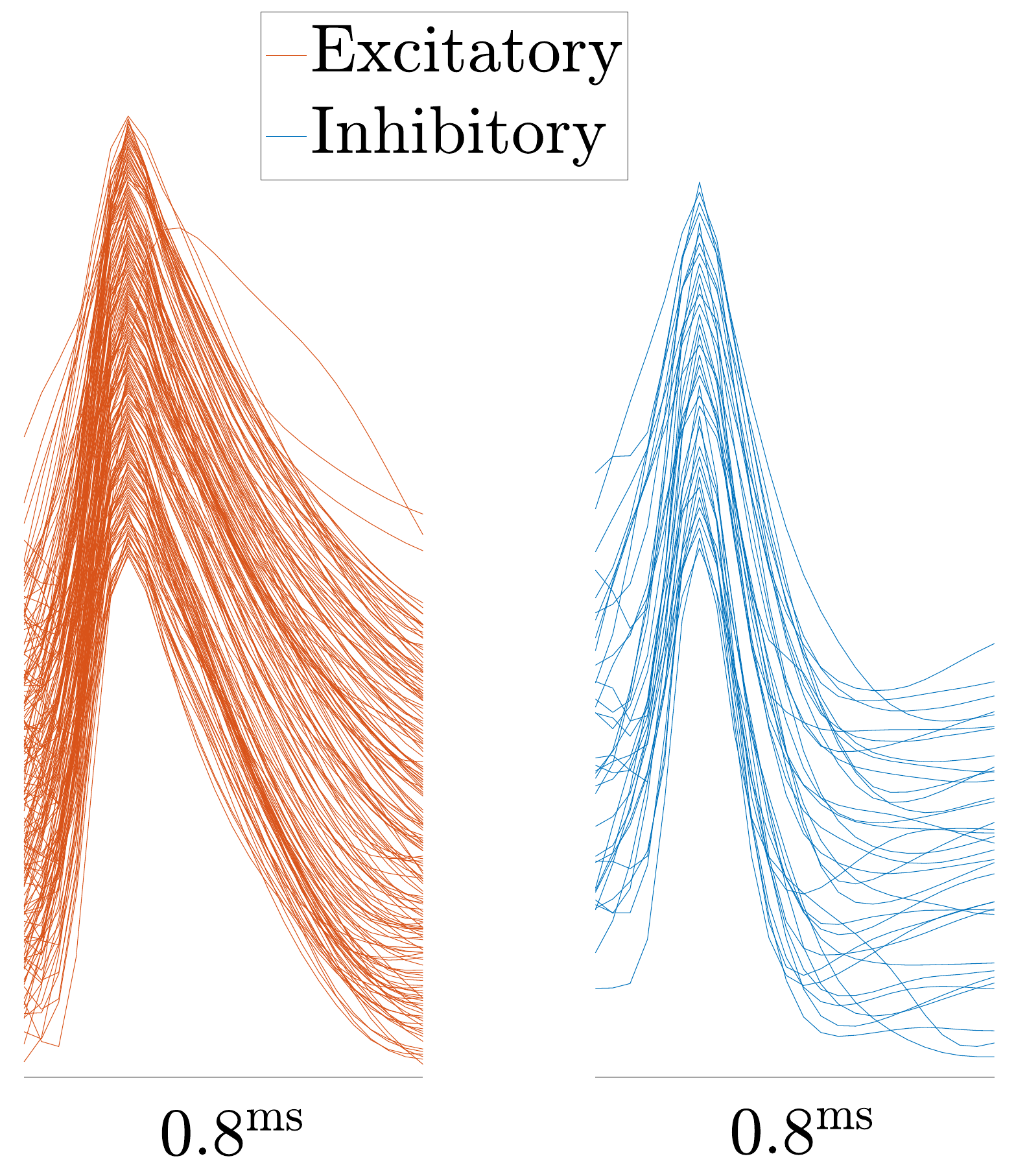}
      }}
    \caption{Excitatory/inhibitory classification of neurons. (a)
      Clusters of spike waveforms. For illustration, clusters are
      shown in the two-dimensional space arising from t-distributed
      stochastic neighbor embedding (t-SNE) dimensionality reduction.
      (b) The spike waveforms of clustered neurons. As expected, the
      inhibitory neurons have faster and narrower spikes.
    }\label{fig:ei}
    \vspace*{-1.5ex}
  \end{figure}

\item Second, we classify the PFC neurons based on their rule-encoding
  (RE) property. This classification was also done
  in~\cite{CCR-MRD:14}, so we briefly review the method for
  completeness. A neuron is said to have a RE property if its firing
  rate is significantly different during the LC and PD trials
  \emph{before the stimulus onset}. 
  In the absence of stimulus, any such difference is attributable to
  the animal's knowledge of the task rule (i.e., which upcoming
  stimulus it has to pay attention to in order to get the
  reward). Thus, it is standard to assess neurons' RE property during
  the \emph{hold period}, namely, the time interval between the
  initiation of each trial and the stimulus onset of that
  trial. Therefore for each PFC neuron, we calculate its mean firing
  rate during the hold period of each trial and then statistically
  compare the results for LC and PD trials ($p < 0.05$, randomization
  test). Among the $112$ neurons in PFC, $40$ encoded for LC while
  $44$ encoded for PD (the remaining PFC neurons with no RE property
  are discarded).

\item Finally, we classify the A1 neurons based on their evoked
  response (ER) property. In contrast to RE, a neuron has an ER
  property if its firing rate is significantly different in response
  to the white noise (LC stimulus) and warble (PD stimulus)
  \emph{after the stimulus onset}. %
  Since the white noise and warble are always presented
  simultaneously, it is not possible to make such a distinction based
  on normal trials. However, before each LC or PD block, the animal is
  only presented with the respective stimulus for a few \emph{cue
    trials} (which is how the animal realizes the rule change). Thus,
  for each A1 neuron, we compare its mean firing rate during the
  \emph{listening period} of each cue trial (namely, the interval
  between the stimulus onset and the time that the animal commits to a
  decision) and statistically compare the distribution of the results
  for LC and PD cue trials ($p < 0.05$, randomization test). Among the
  $99$ A1 neurons, $21$ had an ER for LC while another $21$ had an ER
  for PD (the remaining A1 neurons with no ER property are discarded
  from further analysis).
\end{enumerate}

\begin{remark}\longthmtitle{RE vs. ER detection}
  {\rm It is noteworthy that a smaller fraction of PFC and A1 neurons
    also have ER and RE properties, resp. However, it is expected 
    from systems neuroscience that these properties arise
    from the PFC-A1 interaction, as auditory and
    attention/decision making information disseminate from A1 and PFC,
    resp. This motivates our classification of A1 and PFC
    neurons based on ER and RE, resp., and their reciprocal
    connection in the proposed network structure below. Further, we
    note that our ER detection has a difference with respect
    to~\cite{CCR-MRD:14}. In~\cite{CCR-MRD:14}, the difference between
    the post-stimulus and pre-stimulus firing rates (the latter being
    RE) is used for ER detection, with the motivation of removing the
    portion of post-stimulus firing rate that is due to RE (and thus
    independent of stimulus). However, this relies on the strong
    assumption that the RE and ER responses superimpose linearly,
    which we found likely not to be true based on the statistical
    analysis of the present dataset, perhaps as RE may have driven
    neurons close to their maximum firing rate, leaving little room
    for \emph{additional} ER. We thus use the complete post-stimulus
    firing rate for ER detection, as above. \oprocend}
\end{remark}

As a result of the classifications described above, we group the
neurons into $8$ populations based on the PFC/A1,
excitatory/inhibitory, and LC/PD classifications. The firing rate of
each population (as a function of time) is then calculated as
follows. For each neuron and each trial, the interval $[-10, 10]$
(with time $0$ corresponding to stimulus onset) is decomposed into
$100^\text{ms}$-wide bins and the firing rate of each bin (spike count
divided by bin width) is assigned to the bin's center time. This time
series is then averaged over all trials with the same stimulus pair
and all the neurons within each population, and finally smoothed with
a Gaussian kernel with $1^\text{s}$ standard deviation. This results
in one firing rate time series for each neuron and each stimulus pair.

We limit our choice of stimulus pairs as follows. Recall that each of
LC and PD blocks contains 4 stimulus pairs (left-low, left-high,
right-low, right-high). 
In each block, these 4 pairs are divided into two \emph{go} and two
\emph{no-go} pairs. When the animal hears a go stimulus pair, his
correct response is to go to a nearby food port to receive his
reward. In no-go trials, the correct response is simply inaction
(action is punished by a delay before the animal can do the next
trial). Due to strong motor and reward-consumption artifacts in go
trials (cf.~\cite[Fig. S4]{CCR-MRD:14}), we limit our analysis here to
no-go trials. Further, we also discard the no-go stimulus pair that is
shared between LC and PD blocks, since the correct decision (no-go) is
independent of the block and thus does not require selective
attention. Hence, our analysis only involves one firing rate
time series for each neuronal population in each~block.

\subsection{Network Binary Structure}

We next describe our proposed network binary structure%
\footnote{We here make a distinction between the binary structure of
  the network, composed of only the connectivity pattern among nodes,
  and its full structure, that also includes the connection
  weights.}. In each of the two regions (PFC and A1), the $4$
populations are connected to each other according to the following
physiological properties
(see~\cite{PSG:95,AFTA-MJW-CDP:12,PS-GT-RL-EHB:98}
and~\cite{GKW-RA-BL-HWT-LIZ:08,HKK-SKA-JSI:17,PS-GT-RL-EHB:98} for
evidence of these properties in PFC and A1, resp.):
\begin{enumerate}
\item each excitatory population projects to (i.e., makes synapses on)
  the inhibitory population with the same LC/PD preference (RE in PFC
  or ER in A1);
\item neurons in each excitatory population project to each other
  (captured by the excitatory self-loops in Figure~\ref{fig:struct}).
\item each inhibitory population projects to the populations (both
  excitatory and inhibitory) with \emph{opposite} LC/PD preference
  (the so-called \emph{lateral inhibition} property);
\end{enumerate} 
While within-region connections are both excitatory and inhibitory,
between-region connections in the cortex (including PFC and A1) are
almost entirely excitatory, completing the binary structure shown in
Figure~\ref{fig:struct}.

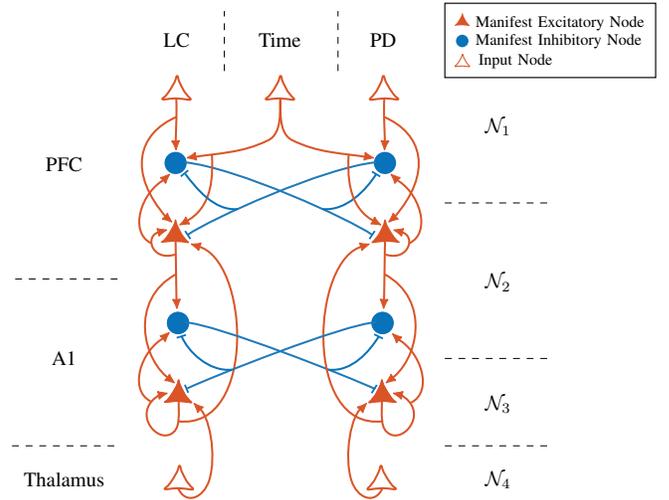
\begin{figure}
  \begin{center} 
  \begin{tikzpicture}
  \node[scale=0.86] (all) {
    \begin{tikzpicture} 
      \node (fig) {\includegraphics[width=0.5\linewidth]{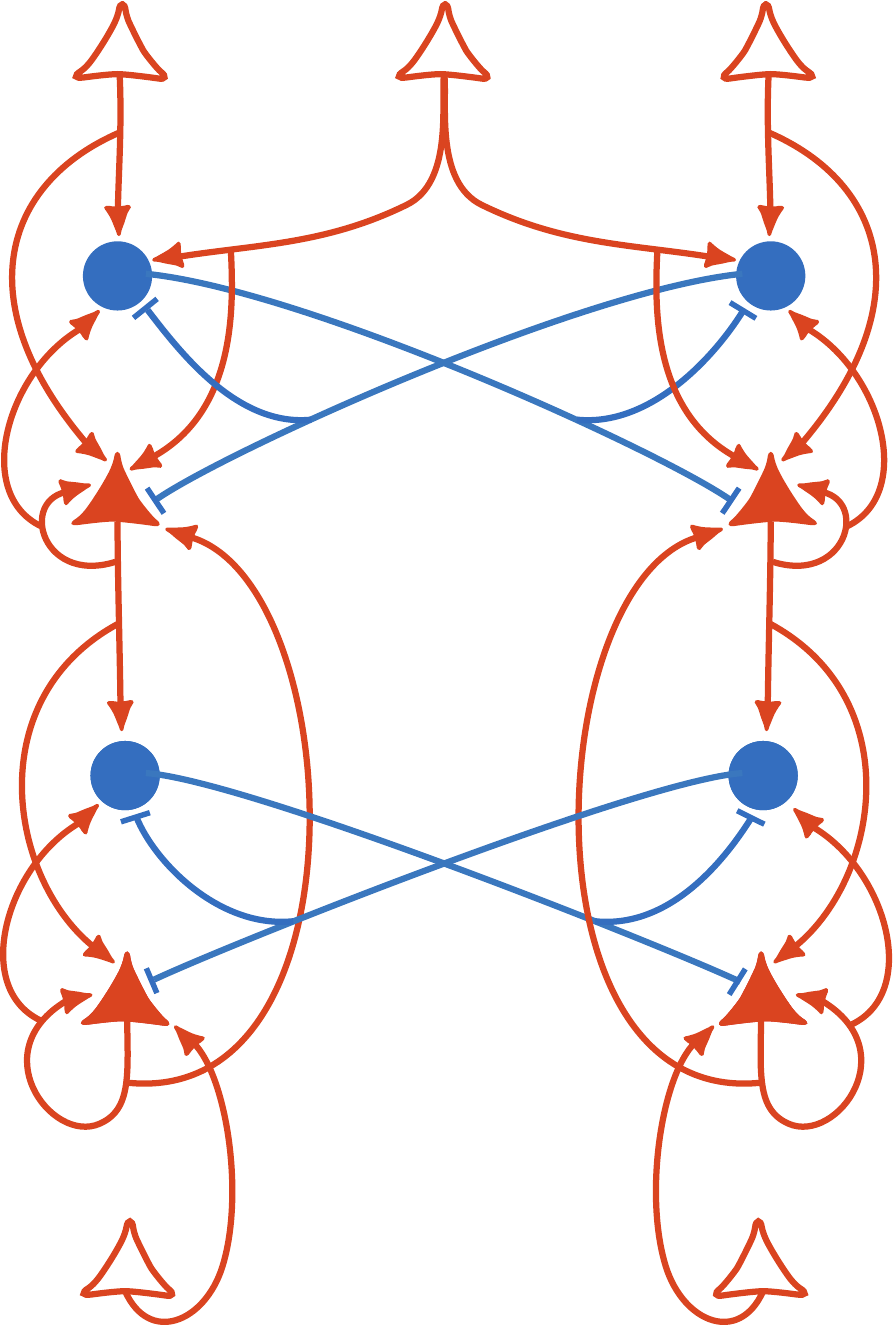}};
      \node[left of=fig, xshift=-40pt, yshift=-70pt] (mgn-r) {};
      \node[left of=mgn-r, xshift=-25pt] (mgn-l) {}; \draw[dashed,
      black] (mgn-l) to (mgn-r); \node[above of=mgn-l, yshift=45pt]
      (a1-r) {}; \node[above of=mgn-r, yshift=45pt] (a1-l) {};
      \draw[dashed, black] (a1-l) to (a1-r); \node[right of=fig,
      xshift=40pt, yshift=-70pt] (n4-l) {}; \node[right of=n4-l,
      xshift=25pt] (n4-r) {}; \draw[dashed, black] (n4-l) to (n4-r);
      \node[above of=n4-l, yshift=10pt] (n3-l) {}; \node[above
      of=n4-r, yshift=10pt] (n3-r) {}; \draw[dashed, black] (n3-l) to
      (n3-r); \node[above of=n3-l, yshift=40pt] (n2-l) {}; \node[above
      of=n3-r, yshift=40pt] (n2-r) {}; \draw[dashed, black] (n2-l) to
      (n2-r); %
      \node[above of=fig, xshift=-25pt, yshift=63pt] (lc-b) {}; %
      \node[above of=lc-b, yshift=5pt] (lc-t) {}; %
      \draw[dashed, black] (lc-b) to (lc-t); %
      \node[above of=fig, xshift=25pt, yshift=63pt] (pd-b) {}; %
      \node[above of=pd-b, yshift=5pt] (pd-t) {}; %
      \draw[dashed, black] (pd-b) to (pd-t); %
      \node[left of=fig, xshift=-67pt, yshift=-85pt, scale=0.9] (mgn) {Thalamus}; %
      \node[above of=mgn, yshift=25pt, scale=0.9] (a1) {A1}; %
      \node[above of=a1, yshift=57pt, scale=0.9] (pfc) {PFC}; %
      \node[right of=fig, xshift=67pt, yshift=-85pt, scale=0.9] (n4) {$\Nc_4$}; %
      \node[above of=n4, yshift=5pt, scale=0.9] (n3) {$\Nc_3$}; %
      \node[above of=n3, yshift=25pt, scale=0.9] (n2) {$\Nc_2$}; %
      \node[above of=n2, yshift=40pt, scale=0.9] (n1) {$\Nc_1$}; %
      \node[above of=fig, xshift=-46pt, yshift=80pt, scale=0.9] (lc) {LC}; %
      \node[right of=lc, xshift=17pt, scale=0.9] (t) {Time}; %
      \node[right of=t, xshift=17pt, scale=0.9] (pd) {PD}; %
      \node[right of=pd, xshift=45pt, yshift=0pt, draw, line
      width=0.5pt, inner sep=3pt, scale=0.7] (leg) {
	\begin{tikzpicture} 
          \node (e) {\includegraphics[width=10pt]{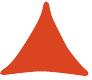}};
          \node[right of=e, xshift=35pt, yshift=-2pt, inner sep=0pt] {\parbox{110pt}{\flushleft Manifest Excitatory Node}};
          \node[below of=e, yshift=15pt] (i) {\includegraphics[width=8pt]{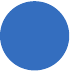}};
          \node[right of=i, xshift=35pt, yshift=3pt, inner sep=0pt] {\parbox{110pt}{\flushleft Manifest Inhibitory Node}};
          \node[below of=i, yshift=16pt] (u) {\includegraphics[width=10pt]{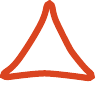}};
          \node[right of=u, xshift=5pt] {Input Node};
	\end{tikzpicture}
	};
    \end{tikzpicture}};
    \end{tikzpicture} 
  \end{center} 
  \caption{Proposed network binary structure. 
    The physiological region, hierarchical layer, and encoding
    properties of nodes are indicated on the left, right, and above
    the figure, resp.}\label{fig:struct}
  \vspace*{-1.5ex}
\end{figure}

\subsubsection*{Hierarchical Structure}
To apply the HSR framework to the network of Figure~\ref{fig:struct},
we still need to assign the nodes to hierarchical layers. This
assignment is in general arbitrary except for two critical
requirements, (i) the existence of timescale separation between layers
and (ii) the existence of both excitatory and inhibitory projections
from any layer to the layer below (to allow for simultaneous
inhibition and recruitment). The trivial choice here is to consider
each region as a layer, which also satisfies (i) (since PFC has slower
dynamics than A1) but not (ii) (since there would be no inhibitory
connection between regions). We thus propose an alternative 3-layer
choice, as shown in Figure~\ref{fig:struct}.%
\footnote{The bottom-most layer $\Nc_4$ represents ``external'' inputs
  from sub-cortical areas. Since we have no recordings from these
  areas, we do not consider any dynamics for $\Nc_4$ and accordingly
  do not include it in HSR analysis.}  This choice clearly satisfies
(ii), and we next show that it also satisfies (i).

\subsubsection*{Computation of Timescales}
To assess the intrinsic timescales of each population, we employ the
common method in neuroscience based on the decay rate of the
correlation
coefficient~\cite{JDM-AB-DJF-RR-JDW-XC-CP-TP-HS-DL-XW:14,RC-KK-MG-HK-XW:15}. In
brief, for each neuron $\ell$, we partition the time window
\emph{before} the stimulus onset%
\footnote{In general, the time interval used for timescale estimation
  should not include stimulus presentation in order to reduce the
  effects of external factors on the internal neuronal dynamics.}
into small bins ($200^\text{ms}$-wide here) and compute the smoothed
mean firing rate of this neuron during each bin and each trial. This
yields a set $\{r_{i, k}^\ell\}_{i, k, \ell}$, where $r_{i, k}^\ell$
denotes the mean firing rate of neuron $\ell$ in the $k$'th time bin
of trial $i$. The Pearson correlation coefficient between two time
bins $k_1$ and $k_2$ is estimated~as
\begin{align*}
  \rho^\ell_{k_1, k_2} = \frac{\sum_i (r_{i, k_1}^\ell - \bar
    r_{k_1}^\ell) (r_{i, k_2}^\ell - \bar r_{k_2}^\ell)}{\sqrt{\sum_i
      (r_{i, k_1}^\ell - \bar r_{k_1}^\ell)^2 \sum_i (r_{i, k_2}^\ell
      - \bar r_{k_2}^\ell)^2}} \in [-1, 1],
\end{align*}
where $\bar r_k^\ell$ is the average of $r_{i, k}^\ell$ across all the
trials for neuron $\ell$. Let $\rho_k^\ell$ be the average of
$\rho_{k_1, k_2}^\ell$ over all $k_1, k_2$ such that $|k_1 - k_2| = k$
and $\bar \rho_k^p$, for any population $p$, be the average of
$\rho_k^\ell$ for all the neurons $\ell$ in the population
$p$. Figure~\ref{fig:taus} shows this function for populations of
excitatory and inhibitory neurons in PFC and A1 (we do not split the
neurons based on their LC/PD preference because it is not relevant for
timescale separation).  Fitting $\bar \rho_k^p$ by an exponential
function of the form $A e^{-k/\tau}$ gives an estimate of the
intrinsic timescale $\tau$ of this population, which becomes exact for
spikes generated by a Poisson point process under certain regularity
conditions~\cite{JDM-AB-DJF-RR-JDW-XC-CP-TP-HS-DL-XW:14}. Here, we use
the range of $k$ values for which the decay of $\bar \rho_k^p$ is
approximately exponential for calculating the fit. As seen in
Figure~\ref{fig:taus}, there is a clear timescale separation between
the layer of A1 excitatory neurons, the layer of A1 inhibitory and PFC
excitatory neurons, and the layer of PFC inhibitory neurons,
satisfying the requirement (i) above.%
\footnote{Note that this method inherently underestimates the
  timescale separation between layers due to the mutual dynamical
  interactions between them.}

\begin{figure}
\centering
  \includegraphics[width=0.9\linewidth]{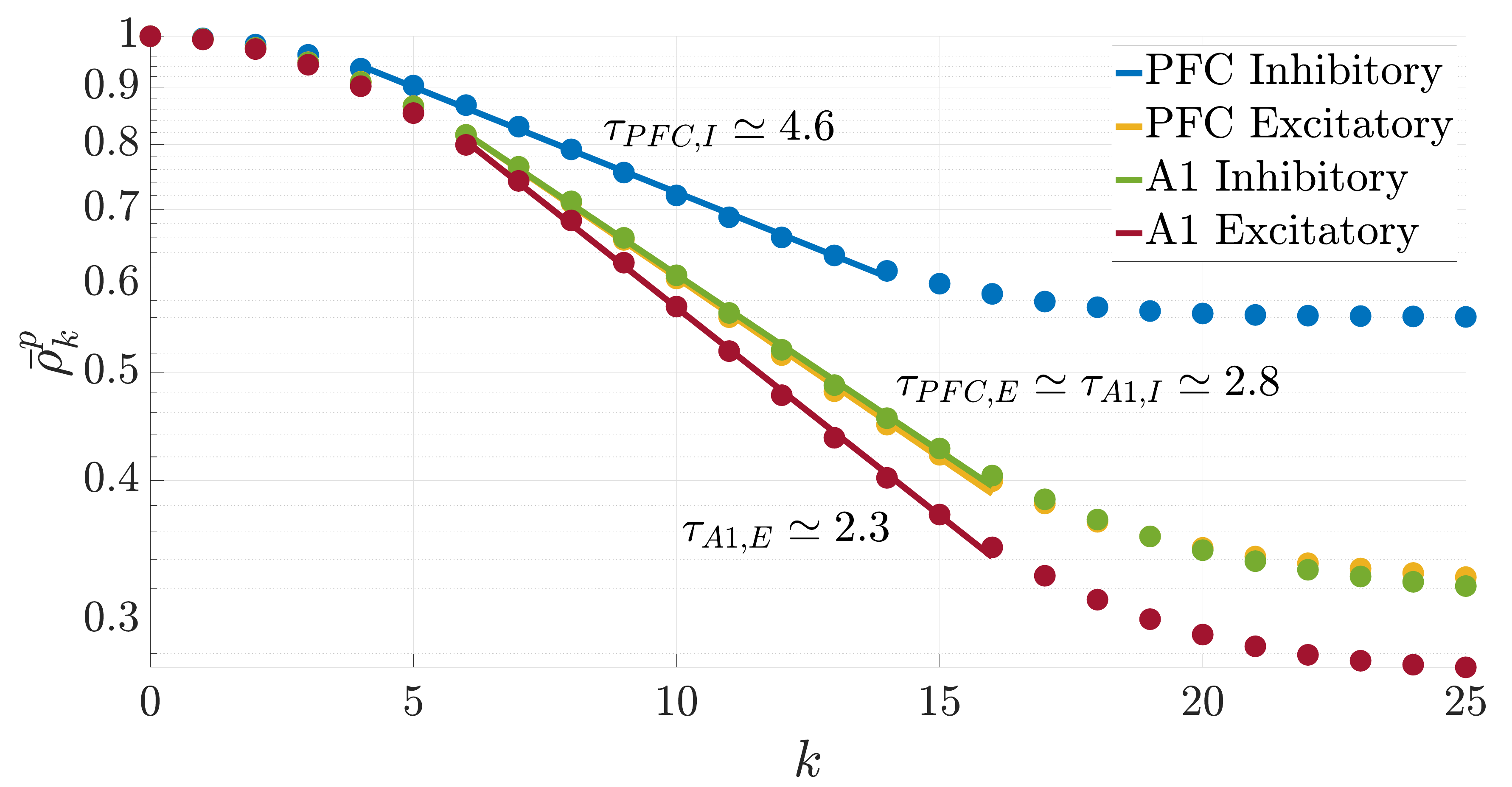}
  \caption{Timescale separation among the layers $\Nc_1$, $\Nc_2$, and
    $\Nc_3$ in Figure~\ref{fig:struct}. The circles illustrate the
    values of the average auto-correlation coefficient $\bar \rho_k^p$
    as a function of time lag $k$, whereas the lines represent the best
    exponential fit over the range of time lags where each $\bar
    \rho_k^p$ decays exponentially (note the logarithmic scale on the
    y-axis).}\label{fig:taus}
\vspace*{-1.5ex}
\end{figure} 

\subsubsection*{Exogenous Inputs and Latent Nodes}
The last step in specifying the binary structure of the network
involves the exogenous inputs to the prescribed neuronal populations
(nodes). Clearly, nodes at the bottom layer (layer $3$) receive
auditory inputs from subcortical areas which we represent as two input
signals $x_1^4$ and $x_2^4$ coming from layer $4$ and corresponding to
the white noise and warble, resp. Both these signals are
constructed by smoothing a square pulse that equals $1$ during
stimulus presentation and $0$ otherwise with the same Gaussian window
used for smoothing the firing rate time-series.

The choice of the inputs to the PFC populations is more intricate. PFC
is itself composed of a complex network of several regions, each
involved in some aspects of high-level cognitive functions. The RE
properties of the recorded PFC populations is only one outcome of such
complex PFC dynamics that also host the animal's overall understanding
of how the task works, his perception of time, etc. In order to
capture the effects of such unrecorded PFC dynamics, we consider 3
additional excitatory PFC populations, as follows. Two input
populations $x_3^1$ and $x_4^1$ simply encode the rule of each block%
\footnote{Note that this static response is different from, and much
  simpler than, the RE of the recorded PFC neurons, which is greatly
  dynamic.}%
:
\begin{align*}
  x_3^1 \equiv \begin{cases}
    1, & \text{if in LC block}, \\
    0, & \text{if in PD block},
  \end{cases} \quad x_4^1 \equiv \begin{cases}
    0, & \text{if in LC block}, \\
    1, & \text{if in PD block}.
  \end{cases}
\end{align*} 
Populations with such a sustained constant activity only as a function
of task parameters are indeed observed during GDSA in
PFC~\cite{NPB-MTH-EMD-RD:15}. The third additional PFC population
encodes the time relative to the stimulus onset, which is critical for
the functioning of the recorded PFC populations. Among the various
forms of encoding time, we consider a population $x_5^1$ with firing
rate
\begin{align*}
  x_5^1(t) = \begin{cases}
    |t_0| - t & t \in [t_0, 0), \\
    0 & t \in (0, t_f],
  \end{cases}
\end{align*}
where $[t_0, t_f] = [-7, 7]$ is the duration of each trial, since
populations with such activity patterns have been observed in
PFC~\cite{AM-HM-KS-YM-JT:09}.%
\footnote{Even though both~\cite{NPB-MTH-EMD-RD:15}
  and~\cite{AM-HM-KS-YM-JT:09} involve primates, populations with
  similar activity patterns are expected to exist in rodents.} %
Since these three populations have very slow dynamics but are
excitatory, following the same logic as before, we position them in
the layer 1 together with the recorded inhibitory PFC populations
$x_1^1$, $x_2^1$.

Finally, to capture the effects of the large populations of neurons
whose activity is not recorded, we consider one \emph{latent} node for
each of the $8$ \emph{manifest} nodes in the network%
\footnote{A node is  \emph{manifest} if its activity is recorded
  during the experiment and \emph{latent} otherwise.} %
with the same in- and out-neighbors as their respective manifest node
(latent nodes are not plotted in Figure~\ref{fig:struct} to avoid
cluttering the network structure). We let $\{x_{1, j}\}_{j = 6, 7}$,
$\{x_{2, j}\}_{j = 5}^8$, and $\{x_{3, j}\}_{j = 3, 4}$ denote these
nodes in $\Nc_1$, $\Nc_2$, and $\Nc_3$, resp.

\subsection{Identification of Network Parameters}

Having established the binary structure of the network, we next seek
to determine its unknown parameters $\Wbf^{i, j}$. While there are
physiological methods for measuring the synaptic weight between a pair
of neurons in vitro, they are not applicable in vivo and thus not
available for our dataset. Also, our nodes consist of several neurons,
making their aggregate synaptic weight an abstract
quantity. Therefore, we resort to system identification/machine
learning techniques to ``learn'' the structure of the network given
its input-output signals.  For this purpose, the choice of objective
function is crucial, for which we propose
\begin{align}\label{eq:f}
  f(z) &= f_\text{SSE}(z) + \gamma_1 f_\text{corr}(z) + \gamma_2
  f_\text{var}(z),
  \\
  \notag f_\text{SSE}(z) &= \sum_{\ell = 1}^2 \sum_{i = 1}^3 \sum_{j =
    1}^{n_{m,i}} \sum_k (\hat x_{i, j}(kT; \ell) - x_{i, j}(kT;
  \ell))^2,
  \\
  \notag f_\text{corr}(z) &= 1 - \frac{1}{2 n_m} \sum_{\ell = 1}^2
  \frac{1}{n_m} \sum_{i = 1}^3 \sum_{j = 1}^{n_{m,i}} \frac{1}{K - 1}
  \\
  \notag &\qquad \times \sum_{k = 1}^K \frac{(\hat x_{i, j}(kT; \ell) - \hat
    \mu_{i, j, \ell}) (x_{i, j}(kT; \ell) - \mu_{i, j, \ell})}{\hat
    \sigma_{i, j, \ell} \sigma_{i, j, \ell}},
  \\
  \notag f_\text{var}(z) &= \Big(\sum_{\ell = 1}^2 \sum_{i = 1}^3
  \sum_{j = 1}^{n_{m,i}} (\hat \sigma_{i, j, \ell} - \sigma_{i, j,
    \ell})^4 \Big)^{1/4},
\end{align}
where,
\begin{itemize}[wide]
\item[--] $z$ is the vector of all unknown network parameters
  consisting of not only the synaptic weights but also the time
  constants $\tau_i$, the background inputs $\cbf_i$, and the initial
  states $\xbf_i(0), i = 1, 2, 3$;
\item[--] $n_{m, i}$ is the number of manifest nodes in layer $i$ (so
  $n_{m, 1} = 2, n_{m, 2} = 4, n_{m, 3} = 2$) and $n_m = 8$ is the
  total number of manifest nodes;
\item[--] $x_{i, j}(t; \ell)$ is the measured state of $j$'th node in
  the $i$'th layer in response to the $\ell$'th stimulus at time $t$
  (where $\ell = 1$ indicates the LC block and $\ell = 2$ the PD
  block) and $\hat x_{i, j}(t; \ell)$ is its model estimate;
\item[--] $T = 0.1$ is the sampling time and $K$ is the total
    number of samples of each signal; and
\item[--] $\mu_{i, j, \ell}, \sigma_{i, j, \ell}, \hat \mu_{i, j,
    \ell}, \hat \sigma_{i, j, \ell}$ are the means and standard
  deviations of $x_{i, j}(\cdot; \ell)$ and $\hat x_{i, j}(\cdot;
  \ell)$, resp.
\end{itemize} 

The rationale behind~\eqref{eq:f} is as follows.  $f_\text{SSE}(z)$ is
the standard sum of squared error (SSE).  In HSR, an important
property of nodal state trajectories is the sign of their derivatives,
which \emph{transiently} indicate recruitment (positive derivative) or
inhibition (negative derivative). This is captured by the average
correlation coefficient $f_\text{corr}(z)$, which is added to
$f_\text{SSE}(z)$ to enforce similar recruitment and inhibition
patterns between measured states and their estimates. Nevertheless,
correlation coefficient between a pair of signals is invariant to the
amount of variation in them, requiring us to add the third term
$f_\text{var}(z)$. The use of $4$-norm in $f_\text{var}(z)$
particularly weights the nodes with large standard deviation
mismatches. Appropriate weights $\gamma_1 = 250$ and $\gamma_2 = 150$
were found via trial and error.

The objective function $f$ is highly nonconvex and we thus use the
\matlab{GlobalSearch} algorithm from the MATLAB Optimization Toolbox
to minimize it. Figure~\ref{fig:x} shows the manifest nodal states as
well as their best model estimates. In order to quantify the
similarity between these states and their estimates, we use the
standard $R^2$ measure given by
\begin{align*}
  R^2 = 1 - \frac{\sum_{\ell, i, j, k} (x_{i, j}(kT; \ell) - \hat
    x_{i, j}(kT; \ell))^2}{\sum_{\ell, i, j, k} (x_{i, j}(kT; \ell) -
    \mu_{i, j, \ell})^2} \simeq 93.6\%.
\end{align*}
This high value is indeed remarkable, especially given the small
network size and the limited availability of measurements in the
experiment ($2240$ data points, $175$ parameters).

\begin{figure}
\centering
  \includegraphics[width=0.9\linewidth]{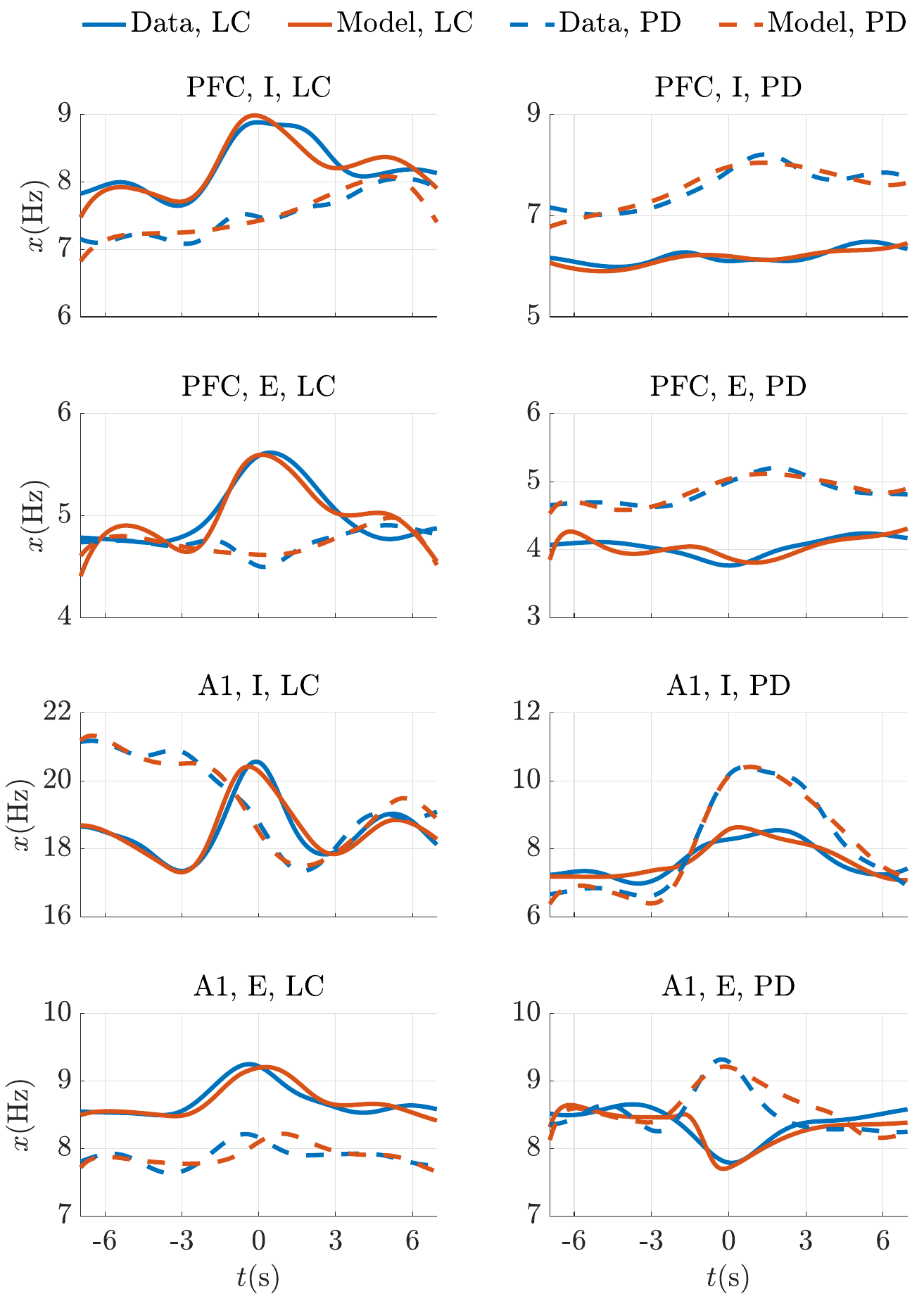}
  \caption{State trajectories of manifest nodes in the network of
    Figure~\ref{fig:struct} (blue: measured, red: model estimate). $t
    = 0$ indicates stimulus onset. Solid and dashed lines correspond
    to LC and PD blocks, resp. The description of each node is
    indicated above its corresponding panel.  The LC/PD in
    the legend refers to the trial rule, while the LC/PD above each
    panel refers to the preference of that particular
    node.}\label{fig:x}
\vspace*{-1.5ex}
\end{figure}

\subsection{Concurrence of the Identified Network with Analysis}

To conclude, we verify here whether the identified network structure
satisfies the requirements of the HSR framework in terms of timescale
separation and stability. Regarding the former, the identified time
constants are given by
\begin{align*}
  \tau_1 = 3.36, \qquad \tau_2 = 1.68, \qquad \tau_3 = 0.70,
\end{align*}
yielding an almost twofold separation of timescales conforming to
Figure~\ref{fig:taus}. Regarding stability, we have to consider the LC
and PD blocks separately (as the definition of task-relevant ($^\ssp$)
and task-irrelevant ($\ssm$) nodes changes according to the block).

In the LC block, the (manifest) LC nodes are task-relevant and the
(manifest) PD nodes are task-irrelevant. Therefore, under this
condition,
\begin{align*}
  W_{3, 3}^{\ssp\ssp} &= 0.01, \qquad \qquad \ \ \Wbf_{3, 2}^{\ssp\ssp}
  = \begin{bmatrix} 0.01 & 0 \end{bmatrix},
  \\
  \Wbf_{2, 2}^{\ssp\ssp} &= \begin{bmatrix} 0.83 & 0 \\ 0.76 &
    0 \end{bmatrix}, \qquad \Wbf_{2, 3}^{\ssp\ssp} =
   \begin{bmatrix} 0.04 \\ 0.58 \end{bmatrix}.
\end{align*}
It is then straightforward to see that
\begin{align*}
  h_3^\ssp(c_3^\ssp) = \begin{cases}
    0 &; \quad c_3^\ssp \le 0 \\
    c_3^\ssp/(1 - W_{3, 3}^{\ssp\ssp}) &; \quad c_3^\ssp \ge 0
  \end{cases} \Rightarrow \bar F_3^\ssp = \frac{1}{1 - W_{3,
      3}^{\ssp\ssp}}.
\end{align*} 
Therefore,
\begin{align*}
  &\rho(|W_{3, 3}^{\ssp\ssp}|) = 0.01 < 1,
  \\
  &\rho\big(|\Wbf_{2, 2}^{\ssp\ssp}| + |\Wbf_{2, 3}^{\ssp\ssp}| \bar
  F_3^\ssp |\Wbf_{3, 2}^{\ssp\ssp}|\big) = \rho\Big(\begin{bmatrix}
    0.83 & 0 \\ 0.77 & 0 \end{bmatrix}\Big) = 0.83 < 1,
\end{align*}
satisfying the sufficient conditions for GES
in~\eqref{eq:rec-cond-sim}. Similarly, in the PD block, we have
\begin{align*}
  &W_{3, 3}^{\ssp\ssp} = 0.01 < 1, \qquad \Wbf_{3, 2}^{\ssp\ssp}
  = \begin{bmatrix} 4.7 \times 10^{-3} & 0 \end{bmatrix},
  \\
  &\Wbf_{2, 2}^{\ssp\ssp} = \begin{bmatrix} 0.12 & 0 \\ 0.56 &
    0 \end{bmatrix}, \qquad \Wbf_{2, 3}^{\ssp\ssp} = \begin{bmatrix}
    0.39 \\ 0.02 \end{bmatrix},
  \\
  &\rho\big(|\Wbf_{2, 2}^{\ssp\ssp}| + |\Wbf_{2, 3}^{\ssp\ssp}| \bar
  F_3^\ssp |\Wbf_{3, 2}^{\ssp\ssp}|\big) = \rho\Big(\begin{bmatrix}
    0.12 & 0 \\ 0.56 & 0 \end{bmatrix}\Big) = 0.12 < 1,
\end{align*}
also satisfying the GES conditions of~\eqref{eq:rec-cond-sim}.

While this concurrence is promising, its robustness to the choice of
dataset and data processing steps is critical. A comprehensive
robustness analysis requires access to multiple datasets and
experimental re-design, which is beyond the scope of this case
study. However, we repeated our entire analysis with
Mann-Whitney-Wilcoxon rank-sum test (used originally
in~\cite{CCR-MRD:14}) and also with varying significance thresholds
$0.001 \le \alpha \le 0.05$ and observed that, despite the change in
the neural populations, our theoretical conditions remained satisfied.

Given the concurrence between the identified network structure and the
hypotheses of our results, Theorems~\ref{thm:sp-inhib}
and~\ref{thm:multi} provide strong analytical support to explain the
conclusions drawn in~\cite{CCR-MRD:14,CCR-MRD:14-crcns} from
experimental data and statistical analysis.  We believe HSR
constitutes a rigorous framework for the analysis of the
multiple-timescale network interactions underlying GDSA, complementing
the conventional statistical and computational analyses in
neuroscience.

\section{Conclusions and Future Work}

We have proposed hierarchical selective recruitment as a framework to
explain several fundamental components of goal-driven selective
attention.  HSR consists of an arbitrary number of neuronal
subnetworks that operate at different timescales and are arranged in a
hierarchy according to their intrinsic timescales.  In this paper, we
have resorted to control-theoretic tools to focus on the top-down
recruitment of the task-relevant nodes. We have derived conditions on
the structure of multi-layer networks guaranteeing the convergence of
the state of the task-relevant nodes in each layer towards their
reference trajectory determined by the layer above in the limit of
maximal timescale separation between the layers. In doing so, we have
characterized the piecewise affinity and global Lipschitzness
properties of the equilibrium maps and unveiled their key role in the
multiple-timescale dynamics of the network.  Combined with the results
of Part I, these contributions provide conditions for the simultaneous
GES of the state of task-irrelevant nodes of all layers to the origin
(inhibition) as well as the GES of the state of task-relevant nodes
towards an equilibrium that moves at a slower timescale as a function
of the state of the subnetwork at the layer above (recruitment).  To
demonstrate that applicability to brain networks, we have presented a
detailed case study of GDSA in rodents and showed that a network with
a binary structure based on HSR and parameters learned using a
carefully designed optimization procedure can achieve remarkable
accuracy in explaining the data while conforming to the theoretical
requirements of HSR.  Our technical treatment has also established a
novel converse Lyapunov theorem for continuous GES switched affine
systems with state-dependent switching.  Future work will include the
extension of this framework to selective inhibition using output
feedback and cases where the recruited subnetworks are asymptotically
stable towards more complex attractors such as limit cycles. Also of
paramount importance is the study of the robustness of network
trajectories as well as the theoretical conditions of HSR to network
parameters, disturbances, and experimental variations (inter-subject
variability, different tasks, measurement noise, etc.). Other topics
of relevance to the understanding of GDSA that we plan to explore are
the analysis of the information transfer along the hierarchy, the
controllability and observability of linear-threshold networks, and
the optimal sensor and actuator placement in hierarchical
interconnections of these networks.

\section*{Acknowledgments}
We would like to thank Dr. Erik J. Peterson for piquing our interest
with his questions on dimensionality control in brain networks and for
introducing us to linear-threshold modeling in neuroscience.  We are
also indebted to Drs. Michael R. DeWeese and Chris C. Rodgers for the
public release of their dataset~\cite{CCR-MRD:14-crcns} and their
subsequent discussion of its details.  This work was supported by NSF
Award CMMI-1826065 (EN and JC) and ARO Award W911NF-18-1-0213~(JC).

\setcounter{section}{0}
\renewcommand{\thesection}{Appendix \Alph{section}}
\section{A Converse Lyapunov Theorem for GES Switched-Affine Systems}\label{app:conv-Lyap}
\renewcommand{\thesection}{\Alph{section}}

The existence of a converse Lyapunov function for
asymptotically/exponentially stable switched linear systems has been
extensively studied for time-dependent switching. Early
works~\cite{APM-YSP:89,WPD-CFM:99} showed that if a switched linear
system is asymptotically (or, equivalently, exponentially) stable
under \emph{arbitrary switching}, then it admits a common Lyapunov
function. This was later extended to infinite-dimensional spaces
in~\cite{FMH-MS:11}. The limitations of these works, however, is the
strong requirement of stability under arbitrary
switching. \cite{FW:05-sicon} proved the existence of a Lyapunov
function under the weaker condition of exponential stability with
minimum dwell-time. Nevertheless, similar results are still missing
for state-dependent switching.  In this appendix, we prove a converse
Lyapunov theorem for continuous GES switched affine systems with
state-dependent switching that is used in both Parts I and II of this
work via~\cite[Lemma~A.2]{EN-JC:18-tacI}. The considered dynamics are
general and subsume the linear-threshold dynamics.

\begin{theorem}\longthmtitle{Converse Lyapunov theorem for GES
    switched affine systems}\label{thm:conv-lyap} 
  Consider the state-dependent switched affine system
  \begin{align}\label{eq:swaf}
    \tau \dot \xbf &= f(\xbf), && \xbf(0) = \xbf_0,
    \\
    \notag f(\xbf) &= \Abf_\lambda \xbf + \bbf_\lambda, \ &&\forall
    \xbf \in \Omega_\lambda = \setdef{\xbf \in D}{\Nbf_\lambda \xbf +
      \pbf_\lambda \le \zeros},
    \\
    \notag & &&\forall \lambda \in \Lambda,
  \end{align}
  where $\Lambda$ is a finite index set, $\Abf_\lambda$ is nonsingular
  for all $\lambda \in \Lambda$, $D = \bigcup_{\lambda \in \Lambda}
  \Omega_\lambda \subseteq \real^n$ is an (open) domain, and
  $\{\Omega_\lambda\}_{\lambda \in \Lambda}$ have mutually disjoint
  interiors. Assume that $f$ is continuous. If~\eqref{eq:swaf} is GES
  towards a unique equilibrium $\xbf^*$, then there exists a
  $C^{\infty}$-function $V: \realnonneg^n \to \real$ and positive
  constants $c_1, c_2, c_3, c_4$ such that for all $\xbf \in D$,
  \begin{subequations} \label{eq:conv-Lyap}
    \begin{align}
      c_1 \|\xbf - \xbf^*\|^2 \le V(\xbf) &\le c_2 \|\xbf -
      \xbf^*\|^2, \label{eq:conv-Lyap1}
      \\
      \frac{\partial V}{\partial \xbf} f &\le -c_3 \|\xbf -
      \xbf^*\|^2, \label{eq:conv-Lyap2}
      \\
      \Big\|\frac{\partial V}{\partial \xbf}\Big\| &\le c_4 \|\xbf -
      \xbf^*\|. \label{eq:conv-Lyap3}
    \end{align}
  \end{subequations}
\end{theorem}
\begin{proof}
  We structure the proof in three steps: (i) showing that the
  solutions of~\eqref{eq:swaf} are continuously differentiable with
  respect to $\xbf_0$ \emph{along its trajectories}, (ii) construction
  of a (not necessarily smooth) Lyapunov-like function that
  satisfies~\eqref{eq:conv-Lyap} along the trajectories
  of~\eqref{eq:swaf}, and (iii) construction of $V$ from this
  Lyapunov-like function (smoothening).
  We only prove the result for
  $\xbf^* = \zeros$ as the general case can be reduced to it with the
  change of variables $\xbf \leftarrow \xbf - \xbf^*$.
  \begin{enumerate}[wide]
  \item Let $\psi(t; \xbf_0)$ denote the unique solution
    of~\eqref{eq:swaf} at time $t \in \real$ (note that we let $t <
    0$). In this step, we prove that $\psi$ is continuously
    differentiable with respect to $\xbf_0$ on $D$ \emph{if} $\xbf_0$
    moves along $\psi$. Precisely, that
    \begin{align}\label{eq:stepi}
      \frac{\partial}{\partial \tau} \psi(t; \psi(\tau; \xbf_0))
      \text{ exists and is continuous at } \tau = 0,
    \end{align}
    for all $\xbf_0 \in D$. First, assume that $\xbf_0 \notin H$,
    where $H \subset D$ is the union of all the switching
    hyperplanes.%
    \footnote{Recall that for each $\lambda$, each row of
      $\Nbf_\lambda \xbf + \pbf_\lambda = \zeros$ defines a switching
      hyperplane.}  Thus, $\xbf_0$ belongs to the interior of a
    switching region, say $\Omega_{\lambda_1}$. Let $\{\lambda_j\}_{j
      = 1}^{J}$, with $ J = J(t) \ge 1$, be the indices of the regions visited by
    $\psi(\tau; \xbf_0)$ during $\tau \in [0, t]$. With a slight abuse
    of notation, let $\Abf_j \triangleq \Abf_{\lambda_j}$ and $\bbf_j
    \triangleq \bbf_{\lambda_j}$, for $j = 1, \dots, J$.  Then,
    \begin{align}\label{eq:psi}
      &\psi(\tau; \xbf_0) =
      \\
      \notag &\!\!\begin{cases} e^{\Abf_1 \tau}(\xbf_0 + \Abf_1^{-1}
        \bbf_1) - \Abf_1^{-1} \bbf_1; &\tau \in [0, t_1],
        \\
        e^{\Abf_2 (\tau - t_1)} (\psi(t_1; \xbf_0) + \Abf_2^{-1}
        \bbf_2) - \Abf_2^{-1} \bbf_2; &\tau \in [t_1, t_2],
        \\
        \quad \vdots
        \\
        e^{\Abf_J (\tau - t_{J\!-\!1})} (\psi(t_{J\!-\!1}; \xbf_0)
        \!+\! \Abf_J^{-1} \bbf_J) \!-\! \Abf_J^{-1} \bbf_J; \!\! &\tau
        \in [t_{J\!-\!1}, t],
      \end{cases}
    \end{align}
    where $t_j = t_j(\xbf_0)$ is the time at which $\psi(\tau;
    \xbf_0)$ crosses the boundary between $\Omega_{\lambda_j}$ and
    $\Omega_{\lambda_{j+1}}$. This expression for $\psi$ is valid for
    all $\xbf$ near $\xbf_0$ that undergo the same sequence of
    switches. To be precise, let $S \subset D$ be the set of
    points lying at the intersection of two or more switching
    hyperplanes and
    \begin{align*}
      S_{(-\infty, 0]} = \setdef{\xbf \in D}{\exists t \in [0, \infty)
        \quad \text{s.t.} \quad \psi(t; \xbf) \in S}.
    \end{align*}
    In words, $S_{(-\infty, 0]}$ is the set of all points that, when
    evolving according to~\eqref{eq:swaf}, will pass through $S$ at
    some point in time. Since $S$ is composed of a finite number of
    affine manifolds of dimensions $n - 2$ or smaller, $S_{(-\infty,
      0]}$ is in turn the union of a finite number of manifolds of
    dimensions $n - 1$ or smaller, and thus has Lebesgue measure zero.

    If $\xbf_0 \notin S_{(-\infty, 0]}$, then it follows from the
    continuity of $\psi$ with respect to $\xbf_0$ on $D$, see
    e.g.,~\cite[Thm 3.5]{HKK:02}, that~\eqref{eq:psi} is valid over a
    sufficiently small neighborhood of $\xbf_0$.  Clearly,
    $\frac{\partial \psi}{\partial \xbf_0}$ then exists and is
    continuous if and only if $t_j$'s are continuously differentiable
    with respect to $\xbf_0$. Consider $t_1$ and let $\nbf^T \xbf + p
    = 0$ be the corresponding switching surface, where $\nbf^T$ is
    equal to some row of $\Nbf_{\lambda_1}$ and equal to minus some
    row of $\Nbf_{\lambda_2}$. $t_1$ is the (smallest) solution to
    \begin{align}\label{eq:t1}
      \nbf^T \big(e^{\Abf_1 \tau}(\xbf_0 + \Abf_1^{-1} \bbf_1) -
      \Abf_1^{-1} \bbf_1\big) + p = 0, \quad \tau \ge 0.
    \end{align}
    The derivative of the lefthand side of~\eqref{eq:t1} with respect
    to $\tau$ equals $\nbf^T f(\psi(t_1; \xbf_0))$, which is nonzero
    if and only if the curve of $\psi$ is not tangent to $\nbf^T \xbf
    + p = 0$. If so, then the continuous differentiability of $t_1$
    with respect to $\xbf_0$ follows from the implicit function
    theorem~\cite{SGK-HRP:02b}. Otherwise, it is not difficult to show
    that $\psi(t; \xbf_0)$ remains in $\Omega_{\lambda_1}$ after
    $t_1$%
    \footnote{This is a general fact about the solutions of linear
      systems and can be shown using the series expansion of the
      matrix exponential.}, %
    contradicting the fact that $t_1$ is a switching time. The same
    argument guarantees that $t_j, j = 2, \dots, J$ are also
    continuously differentiable with respect to $\xbf_0$, and so is
    $\psi(t; \xbf_0)$.

    Before moving on to the case when $\xbf_0 \in S_{(-\infty, 0]}$,
    we analyze the case where still $\xbf_0 \notin S_{(-\infty, 0]}$
    but $\xbf_0 \in H$, i.e., $\xbf_0$ belongs to a switching
    hyperplane, say $\nbf^T \xbf + p = 0$ between $\Omega_{\lambda_1}$
    from $\Omega_{\lambda_2}$, as above. For simplicity, assume $t$ is
    small enough such that $\psi(\tau; \xbf_0)$ remains within
    $\Omega_{\lambda_2}$ for all $\tau \in [0, t]$.%
    \footnote{Note that if $t$ is larger, then subsequent switches to
      $\Omega_{\lambda_j}, j \ge 3$ are similar to the case above
      (where $\xbf_0$ was not on a switching hyperplane) and thus do
      not violate continuous differentiability of $\psi$ with respect
      to $\xbf_0$.}  Let $\xbf$ belong to a sufficiently small
    neighborhood of $\xbf_0$ such that for $\tau \in [0, t]$,
    \begin{align}\label{eq:psix}
      &\psi(\tau; \xbf) =
      \\
      \notag &\!\!\begin{cases} e^{\Abf_2 \tau}(\xbf + \Abf_2^{-1}
        \bbf_2) - \Abf_2^{-1} \bbf_2; &\xbf \in \Omega_{\lambda_2},
        \\
        e^{\Abf_1 \tau}(\xbf + \Abf_1^{-1} \bbf_1) - \Abf_1^{-1}
        \bbf_1; &\xbf \in \Omega_{\lambda_1}, \tau \le t_1,
        \\
        e^{\Abf_2 (\tau - t_1)}(\psi(t_1; \xbf) + \Abf_2^{-1} \bbf_2)
        - \Abf_2^{-1} \bbf_2; \!\!&\xbf \in \Omega_{\lambda_1}, \tau
        \ge t_1,
      \end{cases}
    \end{align}
    where $t_1 = t_1(\xbf)$ is now the solution to $\nbf^T \psi(t_1;
    \xbf) + p = 0$. It is not difficult to show that for $\xbf \in
    \Omega_{\lambda_1}$,
    \begin{align*}
      &\frac{\partial \psi(t; \xbf)}{\partial x_i} = e^{\Abf_2 t}
      \Big[e^{-\Abf_2 t_1} e^{\Abf_1 t_1} e_i + \frac{\partial
        t_1}{\partial x_i}
      \\
      &\qquad \times \Big(\!-\!\Abf_2 e^{-\Abf_2 t_1} e^{\Abf_1 t_1} (\xbf +
      \Abf_1^{-1} \bbf_1) + e^{-\Abf_2 t_1} \Abf_1 e^{\Abf_1 t_1} 
      \\
      &\qquad \times (\xbf + \Abf_1^{-1} \bbf_1)+ \Abf_2 e^{-\Abf_2 t_1} (\Abf_2^{-1} \bbf_2 -
      \Abf_1^{-1} \bbf_1)\Big)\Big],
    \end{align*}
    where $e_i$ is the $i$'th column of $\Ibf_n$. Taking the limit of
    this expression as $\xbf \to \xbf_0$ and using the facts that
    $\lim_{\xbf \to \xbf_0} t_1 = 0$ and $\Abf_1 \xbf_0 + \bbf_1 =
    \Abf_2 \xbf_0 + \bbf_2$, we get
    \begin{align*}
      &\lim_{\ \xbf \stackrel{\Omega_{\lambda_1}}{\to} \xbf_0}
      \!\!\!\! \frac{\partial \psi(t; \xbf)}{\partial x_i} = e^{\Abf_2
        t} e_i, \qquad \forall i \in \until{n},
      \\
      &\Rightarrow \!\!\! \lim_{\ \xbf
        \stackrel{\Omega_{\lambda_1}}{\to} \xbf_0} \!\!\!\!
      \frac{\partial \psi(t; \xbf)}{\partial \xbf} = e^{\Abf_2 t} =
      \!\!\! \lim_{\ \xbf \stackrel{\Omega_{\lambda_2}}{\to} \xbf_0}
      \!\!\!\! \frac{\partial \psi(t; \xbf)}{\partial \xbf},
    \end{align*}
    where the second equality follows directly
    from~\eqref{eq:psix}. Therefore, $\psi(t; \xbf_0)$ is continuously
    differentiable with respect to $\xbf_0$ on the entire $D \setminus
    S_{(-\infty, 0]}$.

    Finally, if $\xbf_0 \in S_{(-\infty, 0]}$, the same expression
    as~\eqref{eq:psi} or~\eqref{eq:psix} (depending on whether $\xbf_0
    \in H$ or not) holds for $\xbf_0$ and also for all $\xbf$ within a
    sufficiently small neighborhood of it \emph{that lie on the same
      system trajectory as $\xbf_0$}. This curve can be parameterized
    in many ways, one of which is given by $\psi(\tau;
    \xbf_0)$. Together with the analysis of the case $\xbf_0 \notin
    S_{(-\infty, 0]}$ above, this proves that \eqref{eq:stepi} exists
    and is continuous at $\tau_0$\footnote{We have indeed proved a
      slightly stronger result than~\eqref{eq:stepi} for $\xbf_0
      \notin S_{(-\infty, 0]}$, which we use in step (ii) below.}.

  \item In this step we introduce a function $\hat V$ that may not be
    smooth but satisfies properties similar
    to~\eqref{eq:conv-Lyap}. Let
    \begin{align*}
      \hat V(\xbf) \triangleq \int_0^\delta \|\psi(t; \xbf)\|^2 d t,
      \qquad \forall \xbf \in D,
    \end{align*}
    where $\delta$ is a constant to be chosen. It is straightforward
    to show that $f$ is globally Lipschitz. Using this and the GES
    of~\eqref{eq:swaf},
    the same argument as in~\cite[Thm 4.14]{HKK:02} shows that
    \begin{align}\label{eq:stepii12}
      2 c_1 \|\xbf\|^2 \le \hat V(\xbf) &\le \frac{2}{3} c_2
      \|\xbf\|^2,
    \end{align}
    for some $c_1, c_2 > 0$. Further, let
    \begin{align*}
      D_{\psi \circ \psi}(t; \tau; \xbf) \triangleq
      \frac{\partial}{\partial \tau} \psi(t; \psi(\tau; \xbf)), \qquad
      t, \tau \in \real, \xbf \in D.
    \end{align*}
    By the definition of $\psi$, we have the identity $ \psi(t; \psi(s
    - t; \xbf)) = \psi(s, \xbf)$, $t, s \in \real, \xbf \in D$.
    Taking $\frac{d}{d t}$ of both sides, we get $ \psi_t(t; \psi(s -
    t; \xbf)) - D_{\psi \circ \psi}(t; s - t; \xbf) = 0$, where
    $\psi_t(t; \xbf) = \frac{\partial \psi(t; \xbf)}{\partial
      t}$. Setting $s = t + \tau$, $ D_{\psi \circ \psi}(t; \tau;
    \xbf) = \psi_t(t; \psi(\tau; \xbf))$.  For the parallel
    of~\eqref{eq:conv-Lyap2}, we then have
    \begin{align*}
      \frac{d}{d \tau} \hat V(\psi(\tau; \xbf)) &= \int_0^\delta 2
      \psi(t; \psi(\tau; \xbf))^T D_{\psi \circ \psi}(t; \tau; \xbf) d
      t
      \\
      & =
      \int_0^\delta 2 \psi(t;
      \psi(\tau; \xbf))^T \psi_t(t; \psi(\tau; \xbf)) d t
      \\
      &= \int_0^\delta \frac{\partial}{\partial t} \|\psi(t;
      \psi(\tau; \xbf))\|^2 d t
      \\
      &= \|\psi(\delta; \psi(\tau; \xbf))\|^2 - \|\psi(\tau;
      \xbf)\|^2.
    \end{align*}
    Thus
    \begin{align}\label{eq:stepii3}
      \frac{d}{d \tau} \hat V(\psi(\tau; \xbf))\Big|_{\tau =
        0} \!\!= \|\psi(\delta; \xbf)\|^2 - \|\xbf\|^2 \le -2 c_3
      \|\xbf\|^2,
    \end{align}
    where the last inequality holds, as shown in~\cite[Thm
    4.14]{HKK:02}, for an appropriate choice of $\delta$ and $c_3 =
    \frac{1}{4}$. Finally, for the parallel of~\eqref{eq:conv-Lyap3},
    recall from step (i) that $\frac{\partial}{\partial \xbf} \psi(t;
    \xbf)$ exists and is continuous on $D \setminus S_{(-\infty,
      0]}$. Therefore, from~\eqref{eq:swaf}, we have
    \begin{align*}
      \frac{\partial}{\partial t} \frac{\partial \psi(t;
        \xbf)}{\partial \xbf} \!=\! \frac{\partial f}{\partial
        \xbf}(\psi(t; \xbf)) \frac{\partial \psi(t; \xbf)}{\partial
        \xbf}, \quad \left.\frac{\partial \psi(t; \xbf)}{\partial
          \xbf}\right|_{t = 0} = \Ibf_n,
    \end{align*}
    on $D \setminus (S_{(-\infty, 0]} \cup H)$. Using the global
    Lipschitzness of $f$ and the fact that $D \setminus S_{(-\infty,
      0]}$ is invariant under~\eqref{eq:swaf}, we have $
    \Big\|\frac{\partial \psi(t; \xbf)}{\partial \xbf}\Big\| \le e^{L
      t}$, for all $x \in D \setminus S_{(-\infty, 0]}$, where $L$ is
    the Lipschitz constant of $f$. The same argument as in~\cite[Thm
    4.14]{HKK:02} then yields
    \begin{align}\label{eq:stepii4}
      \Big\|\frac{\partial \hat V}{\partial \xbf}\Big\| \le
      \frac{2}{3} c_4 \|\xbf\|, \qquad \forall x \in D \setminus
      S_{(-\infty, 0]},
    \end{align}
    for some $c_4 > 0$.

  \item In this step, we follow~\cite[Thm 3 \& 4]{JK:63} to construct
    $V$ as an smooth approximation to $\hat V$ and show that it
    satisfies~\eqref{eq:conv-Lyap}. Since $f$ is globally Lipschitz,
    $\psi(t; \xbf)$ is Lipschitz in $\xbf$ (see, e.g.,~\cite[Ch
    5]{PH:82}) and so is $\hat V$. This, together
    with~\eqref{eq:stepii3}, satisfies all the assumptions
    of~\cite[Thm 4]{JK:63}, which in turn guarantees the existence of
    an infinitely smooth $V$ such that
    \begin{subequations}
      \begin{align}
        |V(\xbf) - \hat V(\xbf)| &< \frac{1}{2} \hat V(\xbf), \qquad
        \forall \xbf \in D,
        \\
        \frac{\partial V}{\partial \xbf} f(\xbf) &< -c_3
        \|\xbf\|^2, \label{eq:Kurzweil2}
      \end{align}
    \end{subequations}
    for all $\xbf \in D$. Equation~\eqref{eq:conv-Lyap1} follows
    immediately from~\eqref{eq:Kurzweil2} and~\eqref{eq:stepii12}. To
    prove~\eqref{eq:conv-Lyap3}, we note that the same construction of
    $V$ as in~\cite[Thm 3 \& 4]{JK:63} satisfies
    \begin{align*}
      \Big\|\frac{\partial V}{\partial \xbf} - \frac{\partial \hat
        V}{\partial \xbf}\Big\| &< \frac{1}{2} \Big\|\frac{\partial
        \hat V}{\partial \xbf}\Big\|, \qquad \forall \xbf \in D
      \setminus S_{(-\infty, 0]},
    \end{align*}
    if the constants $\xi_{i, k}$ and $\zeta_{i, k}, i, k = \dots, -2,
    0, 2, \dots$ (and consequently the corresponding $\bar r_{i, k},
    i, k = \dots, -2, 0, 2, \dots$) are chosen sufficiently
    small. This, together with~\eqref{eq:stepii4},
    guarantees~\eqref{eq:conv-Lyap3}, completing the proof.
  \end{enumerate}
  \vspace*{-15pt}
\end{proof}

\renewcommand{\thesection}{Appendix \Alph{section}}
\section{Additional Proofs}\label{app:pf}
\renewcommand{\thesection}{\Alph{section}}

\begin{proof}[Proof of Lemma~\ref{lem:affinity}]
  Pick $\cbf' \in \real^{n'}$ and let $\xbf^*$ be the unique solution
  of~\eqref{eq:gen-eq}. Since $\bigcup_{\lambda \in \Lambda}
  \Psi_\lambda = \real^n$, let  $\lambda \in \Lambda$ with
  \begin{align}\label{eq:Psi-lambda}
    \Wbf_3 \xbf^* + \bar \cbf \in \Psi_\lambda.
  \end{align}
  If $\Wbf_3 \xbf^* + \bar \cbf$ lies on the boundary of more than one
  $\Psi_\lambda$, pick one arbitrarily. Therefore, $\xbf^*$ satisfies
  \begin{align*}
    \xbf^* &= [(\Wbf_1 + \Wbf_2 \Fbf_\lambda \Wbf_3) \xbf^* + \Wbf_2
    (\Fbf_\lambda \bar \cbf + \fbf_\lambda) + \cbf']_\zeros^\mbf.
  \end{align*}
  From~\eqref{eq:h-pa}, it follows that $h'$ has the
  form~\eqref{eq:h'} with $\lambda' \triangleq (\lambda, \sigmab)$ and
  $\Lambda' = \Lambda \times \zls^{n'}$. The quantities
  $\Fbf'_{\lambda'}, \fbf'_{\lambda'}, \Gbf'_{\lambda'},
  \gbf'_{\lambda'}$ also have the same form as in~\eqref{eq:h-pa}
  except that here
  \begin{align*}
    \Wbf &= \Wbf_1 + \Wbf_2 \Fbf_\lambda \Wbf_3,
    \\
    \fbf'_{\lambda'} &= (\Ibf - \Sigmab^\ell \Wbf)^{-1} \Sigmab^\s \mbf + (\Ibf - \Sigmab^\ell \Wbf)^{-1} \Sigmab^\ell \Wbf_2 (\Fbf_\lambda \bar \cbf + \fbf_\lambda).
  \end{align*} 
  The proof is complete noting that $\bigcup_{\lambda' \in
    \Lambda'} \Psi'_{\lambda'} = \real^{n'}$ since any $\cbf' \in
  \real^{n'}$ must be in at least one $\Psi'_{\lambda'}$ by
  construction.
\end{proof}

\hspace{0.75em} \emph{Proof of Lemma~\ref{lem:h-lip-gen}:}
  Pick any $\cbf, \hat \cbf \in \real^n$. Since all the sets
  $\Psi_\lambda$ are convex, the line segment $\gamma \triangleq
  \setdefb{\big(\theta, (1 - \theta) \cbf + \theta
    \hat \cbf\big)}{\theta \in [0, 1]}$ joining $\cbf$ and $\hat \cbf$ can
  be broken into $k \le |\Lambda| < \infty$ pieces such that $\gamma =
  \bigcup_{i = 1}^k \gamma_i, \gamma_i \triangleq
  \setdefb{\big(\theta, (1 - \theta) \cbf + \theta
    \hat \cbf\big)}{\theta \in [\theta_{i - 1}, \theta_i]}, \theta_0 = 0,
  \theta_k = 1$ and each $\gamma_i \subset \Psi_{\lambda_i}$ for some
  $\lambda_i \in \Lambda$. Let $\cbf_i \triangleq (1 - \theta_i)
  \cbf + \theta_i \hat \cbf$. Then,
  \begin{align*}
    &\|h(\cbf) - h(\hat \cbf)\| = \Big\|\sum_{i = 1}^k
      \big(h(\cbf_{i-1}) - h(\cbf_i)\big) \Big\|
    \\
    &\le \sum_{i = 1}^k \|h(\cbf_{i-1}) - h(\cbf_i)\| = \sum_{i = 1}^k
    \|\Fbf_{\lambda_i}(\cbf_{i-1} - \cbf_i)\|
    \\
    &\le \Big[\!\max_{\lambda \in \Lambda}\|\Fbf_\lambda\| \Big]
    \sum_{i = 1}^k \|\cbf_{i-1} - \cbf_i\| = \Big[\!\max_{\lambda \in
      \Lambda}\|\Fbf_\lambda\| \Big] \|\cbf - \hat \cbf\|. \quad \QED
  \end{align*} 

\vspace*{-3ex}

\begin{IEEEbiography}[{\includegraphics[width=1in, height=1.25in, clip, keepaspectratio]{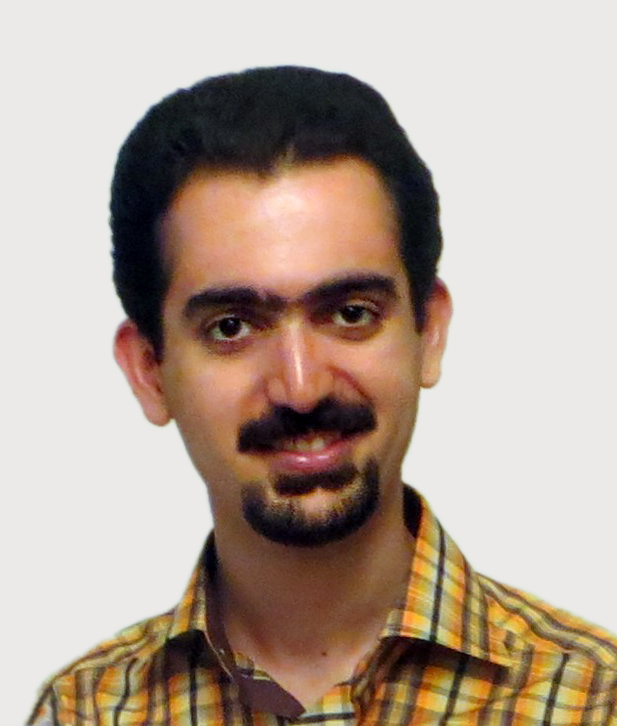}}]{Erfan Nozari}
  received his B.Sc. degree in Electrical Engineering-Control in 2013
  from Isfahan University of Technology, Iran and Ph.D. in Mechanical
  Engineering and Cognitive Science in 2019 from University of
  California San Diego. He is currently a postdoctoral researcher at
  the University of Pennsylvania Department of Electrical and Systems
  Engineering. He has been the (co)recipient of the 2019 IEEE
  Transactions on Control of Network Systems Outstanding Paper Award,
  the Best Student Paper Award from the 57th IEEE Conference on
  Decision and Control, the Best Student Paper Award from the 2018
  American Control Conference, and the Mechanical and Aerospace
  Engineering Distinguished Fellowship Award from the University of
  California San Diego. His research interests include dynamical
  systems and control theory and its applications in computational and
  theoretical neuroscience and complex network systems.
\end{IEEEbiography}

\vspace*{-3ex}

\begin{IEEEbiography}
  [{\includegraphics[width=1in,height=1.25in,clip,keepaspectratio]{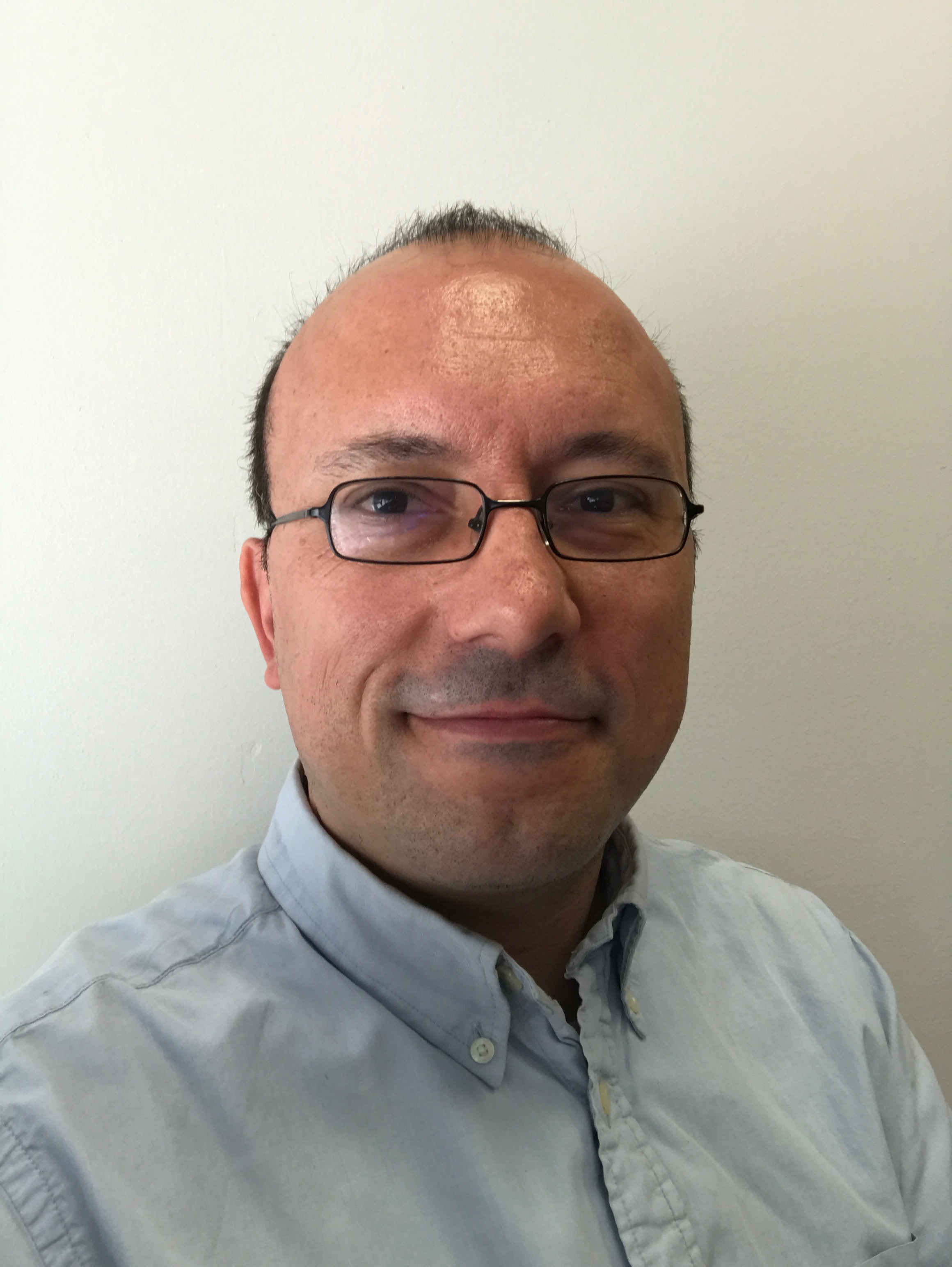}}]
  {Jorge Cort\'es} (M'02-SM'06-F'14) received the Licenciatura degree
  in mathematics from Universidad de Zaragoza, Spain, in 1997, and the
  Ph.D. degree in engineering mathematics from Universidad Carlos III
  de Madrid, Spain, in 2001. He held postdoctoral positions with the
  University of Twente, The Netherlands, and the University of
  Illinois at Urbana-Champaign, USA. He was an Assistant Professor
  with the Department of Applied Mathematics and Statistics,
  University of California, Santa Cruz, USA, from 2004 to 2007. He is
  currently a Professor in the Department of Mechanical and Aerospace
  Engineering, University of California, San Diego, USA. He is the
  author of Geometric, Control and Numerical Aspects of Nonholonomic
  Systems (Springer-Verlag, 2002) and co-author (together with
  F. Bullo and S. Mart{\'\i}nez) of Distributed Control of Robotic
  Networks (Princeton University Press, 2009).  His current research
  interests include distributed control and optimization, network
  science, resource-aware control, decision making under uncertainty,
  and distributed coordination in power networks, robotics, and
  transportation.
\end{IEEEbiography}

\end{document}